\documentclass{article}
\usepackage{amsmath,amssymb,amsthm,geometry,microtype,enumitem}
\usepackage{natbib}

\usepackage{tikz}
\usepackage{pgf} 
\usetikzlibrary{patterns,automata,arrows,shapes,snakes,topaths,trees,backgrounds,positioning,through,calc}
\theoremstyle{definition}
\newtheorem{theorem}{Theorem}[section]
\newtheorem{proposition}[theorem]{Proposition}

\newtheorem{remark}[theorem]{Remark}

\definecolor{medgreen}{rgb}{0.0, 0.75, 0.0}

\usepackage{comment}

\usepackage{hyperref}
\hypersetup{colorlinks=true, citecolor=blue, linkcolor=blue}  

\setcitestyle{round}

\usepackage{setspace}

\begin{document}
\title{Axiomatizations of \\ a simple Condorcet voting method \\ for Final Four and Final Five elections}
\author{Wesley H. Holliday\\ {\small University of California, Berkeley}}
\date{September 13, 2025}

\onehalfspace

\maketitle

\begin{abstract}
Proponents of Condorcet voting face the question of what to do in the rare case when no Condorcet winner exists. Recent work provides compelling arguments for the rule that should be applied in three-candidate elections, but already with four candidates,  many rules appear reasonable. In this paper, we consider a recent proposal of a simple Condorcet voting method for Final Four political elections. Our question is what normative principles could support this simple form of Condorcet voting. When there is no Condorcet winner, one natural principle is to pick the candidate who is \textit{closest} to being a Condorcet winner. Yet there are multiple plausible ways to define closeness, leading to different results. Here we take the following approach: identify a relatively uncontroversial sufficient condition for one candidate to be closer than another to being a Condorcet winner; then use other principles to help settle who wins in cases when that condition alone does not. We prove that our principles uniquely characterize the simple Condorcet voting method for Final Four~elections. This analysis also points to a new way of extending the method to elections with five or more candidates that is simpler than an extension previously considered. The new proposal is to elect the candidate with the most head-to-head wins, and if multiple candidates tie for the most wins, then elect the one who has the smallest head-to-head loss. We provide additional principles sufficient to characterize this simple method for Final Five elections.
\end{abstract}

\section{Introduction}\label{Intro}

Condorcet voting methods are based on two simple ideas, dating back to Condorcet \citeyearpar{Condorcet1785}. First, voters should be allowed to rank the candidates in an election. Second, if one candidate defeats each of the other candidates in head-to-head majority comparisons based on the voters' rankings, then that candidate---called the \textit{Condorcet winner}---should win the election. Condorcet voting methods differ from each other only in their plan for what to do in the rare case when no Condorcet winner exists.\footnote{Out of over 300 single-winner U.S. political elections with ranked ballots and at least three non-write-in candidates documented by \citealt{Otis2022}, we know of only two cases in which a Condorcet winner did not exist (see \citealt[Note 7]{Holliday2025}).} Many plans have been proposed, some of which are quite sophisticated. As a result of the innovativeness of voting theorists in developing different Condorcet methods (see, e.g., \citealt{Daunou1803}, \citealt{Nanson1882}, \citealt{Baldwin1926}, \citealt{Kemeny1959},  \citealt{Tideman1987}, \citealt{Schulze2011}, \citealt{HP2023,HP2023b}, \citealt{Doring2025}), there is now a caricature of Condorcet voting methods as highly complicated. Poundstone \citeyearpar[p.~223]{Poundstone2008} paints this picture amusingly:
\begin{quote} Maybe the biggest hitch with Condorcet voting is the simplicity issue. Condorcet supporters talk up the populist appeal of the John Wayne, last-man-standing premise. Soccer moms and NASCAR dads will nod their heads to that. The thing is, it's impossible to explain Condorcet voting any further without talking about cycles and “beat-paths” and “cloneproof Schwartz sequential dropping.” This is where the android's mask falls off and it's all wires and microchips inside. Non-Ph.D.s run screaming for the exits. So far, Condorcet voting has tended to appeal to the kind of people who can write JavaScript code for it.\end{quote}
For uses of voting methods by clubs, committees, professional associations, etc., there may be no objection to the use of sophisticated Condorcet methods, but for \textit{political} uses of voting methods, simplicity is key. In recent years, a number of scholars have argued for the use of Condorcet voting in political elections in the United States (see, e.g., \citealt{MaskinDasgupta2004}, \citealt{Maskin2017a,Maskin2017b}, \citealt{Foley2022,Foley2024}, \citealt{Foley2021,Foley2023}, \citealt{Atkinson2024}, \citealt{Robinette2024}). Fortunately, contrary to the impression given by Poundstone, there are Condorcet voting methods that are easy to explain to voters and policymakers. Consider the following method, which does not require a Ph.D. to understand:
\begin{itemize}
\item Proposed method: elect the candidate with the most head-to-head wins; if multiple candidates tie for the most wins, elect the one who has the smallest head-to-head loss in terms of margin of votes.
\end{itemize}
This method was proposed in \citealt{Holliday2025} for Final Four general elections, of the kind held in Alaska (\citealt{Brooks2020}), which feature four candidates chosen in a previous election to compete in the general election. The method was motivated by its simplicity, by its ability (like other Condorcet methods) to elect candidates who are highly representative of the electorate (\citealt{HP2024SocUtil}), and by its satisfaction of several desirable axioms from voting theory, such as Immunity to Spoilers (\citealt{HP2023}). 

There is another kind of motivation for a voting method that can be given: a complete  \textit{axiomatic characterization} showing that the given method is the \textit{only} possible voting method satisfying some list of axioms. Insofar as the axioms formalize \textit{desirable} properties of a voting method, such a characterization helps justify the use of the method. Of course, such a characterization does not decisively settle that one should use the voting method in question, since other axiomatic characterizations of other voting methods highlight the strengths of those other methods.\footnote{For axiomatic characterizations of other ranked voting methods, see, e.g., \citealt{Young1974}, \citealt{Henriet1985}, \citealt{Merlin2003}, \citealt{Freemanetal2014}, and \citealt{Ding2022}.} But it does help one understand the unique package of benefits of using the method in question. Recently, Condorcet's \citeyearpar{Condorcet1785} own voting method for three-candidate elections was axiomatically characterized in \citealt{HP2025}, as an extension of May's \citeyearpar{May1952} famous axiomatic characterization of majority rule in two-candidate elections.\footnote{One of the axioms that Holliday and Pacuit \citeyearpar{HP2025} add to May's axioms is the axiom of Immunity to Spoilers mentioned above, which we will use in Section~\ref{Five}. For a different axiomatization of Condorcet's method for three-candidate elections, not based on May's axioms, see \citealt{Brandt2025}.} Many Condorcet methods, including the proposed method above (as well as those in \citealt{Simpson1969}, \citealt{Kramer1977}, \citealt{Tideman1987}, \citealt{Schulze2011}, \citealt{HP2023,HP2023b}, \citealt{Doring2025}), are equivalent to Condorcet's method when restricted to three-candidate elections (ignoring head-to-head ties): simply elect the Condorcet winner, if there is one, and otherwise elect the candidate with the smallest loss. This convergence in the case of three candidates, as well as axiomatic considerations (\citealt{HP2025}, \citealt{Brandt2025}), provides compelling arguments for the use of Condorcet's method in three-candidate elections. However, as soon as we move to four candidates, the various Condorcet methods diverge, which motivates axiomatic analyses of their differences.

In this paper, we provide an axiomatic characterization of the proposed method above for Final Four elections. When there is no Condorcet winner in an election, one natural principle is to pick the candidate who is \textit{closest} to being a Condorcet winner. Yet there are multiple plausible ways to define closeness, leading to different results (\citealt{Elkind2016}). Here we take the following approach: identify a relatively uncontroversial sufficient condition for one candidate to be closer than another to being a Condorcet winner; then use other principles to help settle who wins in cases when that condition alone does not. We prove that our principles uniquely characterize the proposed method for Final Four~elections. 

Our analysis also points to a way of extending the method to elections with five or more candidates that is different from the extension briefly contemplated in \citealt[\S~2.3]{Holliday2025}. For a four-candidate election, there is always at least one candidate with at most one loss,\footnote{\label{AtMostOne}If each candidate had two losses, then there would be eight total losses, but there are only six head-to-head matches between four candidates.} whereas in a five-candidate election, it is possible for every candidate to suffer two losses (see Figure~\ref{Penta}), which leads to a choice point. In \citealt[\S~2.3]{Holliday2025}, the proposed method for Final Four elections was generalized to elections with five or more candidates in such a way that if multiple candidates tie for the most wins, then the one whose \textit{worst loss} is \textit{smallest} is elected. While that is a natural idea, borrowed from the Minimax voting method (\citealt{Simpson1969}, \citealt{Kramer1977}), here we will give a motivation for the simpler definition in the bullet point above: among the candidates tied for the most wins, we look at each candidate's smallest loss (not their worst loss) and elect the candidate with the smallest loss.\footnote{Although we will not go into the subtleties of further tie-breaking in this paper, here are two reasonable rules for handling extremely rare cases: first, a head-to-head tie between two candidates (i.e., exactly the same number of voters rank candidate $A$ above candidate $B$ as rank $B$ above $A$) counts as 1/2 of a win for each of the two candidates; second, if among the candidates tied for the most wins, there are multiple candidates with the smallest loss, then we break this tie in favor of the one whose \textit{second}-smallest loss is smallest, and so on. If even such lexicographic comparison of the candidates' losses does not break the tie, then a random device could be used to select from among the still-tied candidates.}  The basic idea is that among the candidates tied for the most wins, the one with the smallest loss is \textit{closest} to being an outright winner based on their number of wins, since all it would take is flipping that small loss. We will give a complete axiomatic characterization of this simple extension of the proposed method to Final Five elections by adding two axioms, one of which is the Immunity to Spoilers axiom originally used to motivate the proposed method. 

It would of course be desirable to axiomatically characterize the proposed method for any number of candidates. But since efforts to implement Final Four or Final Five general elections with ranked ballots currently have the most momentum (see \citealt{Gehl2020}, \citealt{Gehl2023}), these are the most important cases to study for relevance to near-term political applications of Condorcet voting methods.

Our axiomatic analysis also helps to distinguish the proposed method from another voting method introduced in \citealt{Holliday2025} (Notes 11 and 22):
\begin{itemize}
\item Variant method: elect the candidate with the most head-to-head wins; if multiple candidates tie for the most wins, then among the candidates in that group, elect the one who beats each of the others in that group head-to-head; if there is no such candidate, then elect the candidate in that group with the smallest head-to-head loss to another candidate in that group.
\end{itemize}
In the case of Final Four elections, this variant method delivers a different result from the proposed method only in an extremely rare type of election,\footnote{For justification of this claim of rarity, see Table~1 in \citealt{Holliday2025}.} dubbed the LS four cycle, which we will discuss below. In \citealt{Holliday2025}, Note 22, it was claimed that ``One can give normative arguments in favor of the proposed method's choice in the LS four cycle, but the greater simplicity of the proposed method's definition and the improbability of the LS four cycle seem already sufficient for favoring the proposed method''  over the variant method. In this paper, we will identify such normative arguments in terms of axioms satisfied by the proposed method but violated by the variant method, namely the axioms that we call Proximity to Condorcet, Independence of Irrelevant Defeats, and Immunity to Spoilers.

The rest of the paper is organized as follows. In Section~\ref{Prelim}, we review some preliminaries from voting theory (for an introduction, see \citealt{Zwicker2016} or \citealt{Pacuit2019}). In Section~\ref{Axioms}, we introduce the axioms that will uniquely characterize the proposed method for Final Four elections. We prove this characterization result in Section~\ref{Proof}. By adding two more axioms, we uniquely characterize the proposed method also in the case of Final Five elections in Section~\ref{Five}. We conclude in Section~\ref{Conclusion}.

\section{Preliminaries}\label{Prelim}

Given a finite set $\mathsf{V}$ of voters and a finite set $\mathsf{C}$ of candidates, a \textit{preference profile} $\mathbf{P}$ for $\mathsf{V}$ and $\mathsf{C}$ assigns to each voter in $\mathsf{V}$ a ranking of $\mathsf{C}$.\footnote{Formally, by a `ranking' we mean a complete and transitive binary relation on $\mathsf{C}$.} In political elections with ranked ballots, typically a voter is not required to rank all the candidates on their ballot. But a ballot with unranked candidates can still be represented in a preference profile by interpreting all unranked candidates as tied at the bottom of the voter's ranking. Since political elections typically do not allow ties among the ranked candidates (at least in the U.S.), a ranked ballot can be modeled by what Ding et al.~\citeyearpar{Ding2025} call a \textit{linear order with bottom indifference}, i.e., a ranking without ties followed by one class of tied candidates at the bottom, representing the unranked candidates.

A \textit{voting method} is a function $\mathcal{F}$ whose domain is a set of profiles such that for any profile $\mathbf{P}$ in the domain, $\mathcal{F}(\mathbf{P})$ is a nonempty subset of the set of candidates in $\mathbf{P}$. Ideally $\mathcal{F}(\mathbf{P})$ would include only one candidate---a unique winner---but due to the possibility of ties, we allow $\mathcal{F}(\mathbf{P})$ to be a set of multiple candidates.

\subsection{Margin-based voting methods}\label{MarginBased}

Given two candidates $A,B$ in $\mathsf{C}$, we say that the \textit{margin of $A$ vs.~$B$} in a preference profile is the number of voters who rank $A$ strictly above $B$ minus the number of voters who rank $B$ strictly above $A$. If the margin is positive, then $A$ \textit{beats $B$ head-to-head}, whereas if the margin is negative, then $A$ \textit{loses to $B$ head-to-head}. A \textit{Condorcet winner} is a candidate who beats every other candidate head-to-head. A voting method satisfies the Condorcet Criterion and is called a \textit{Condorcet method} if it selects the Condorcet winner, and only the Condorcet winner, in any profile in its domain that has a Condorcet winner.

A voting method $\mathcal{F}$ is \textit{margin-based} (also called \textit{pairwise} in \citealt{Brandt2018a} and \textit{C1.5} in \citealt{DeDonder2000}) if whenever two profiles $\mathbf{P}$ and $\mathbf{P}'$ give rise to exactly the same head-to-head margins between candidates, then $\mathcal{F}$ selects the same winners in both profiles: $\mathcal{F}(\mathbf{P})=\mathcal{F}(\mathbf{P}')$. Thus, a margin-based voting method only needs to know the head-to-head margins between candidates in order to determine which candidates win. Many Condorcet voting methods are margin-based (e.g., those defined in \citealt{Copeland1951}, \citealt{Simpson1969}, \citealt{Kramer1977}, \citealt{Tideman1987}, \citealt{Schulze2011}, \citealt{HP2023,HP2023b}, \citealt{Doring2025}).\footnote{An example of a non-margin-based Condorcet method is the ``Condorcet-Hare'' or ``Condorcet-IRV'' method that elects a Condorcet winner, if one exists, and otherwise uses Instant Runoff Voting to select the winner (see \citealt{Green2016}).}

In this paper, we restrict our attention to margin-based voting methods, so it is important that we provide some motivation for doing so. Recently Ding et al.~\citeyearpar{Ding2025} identified normative principles that are together equivalent to a voting method being margin-based. A voting method $\mathcal{F}$ satisfies Preferential Equality if for any preference profile $\mathbf{P}$ and candidates $A,B$ in $\mathbf{P}$, if $\mathbf{P}^i$ is obtained from $\mathbf{P}$ by just one voter $i$ changing their ballot from ranking $A$ immediately above $B$ to ranking $B$ immediately above $A$, and similarly $\mathbf{P}^j$ is obtained from $\mathbf{P}$ by just one voter $j$ changing their ballot from ranking $A$ immediately above $B$ to ranking $B$ immediately above $A$,\footnote{Assume that neither $A$ nor $B$ is in a tie with any other candidate in $i$'s ranking or in $j$'s ranking.} then $\mathcal{F}(\mathbf{P}^i)=\mathcal{F}(\mathbf{P}^j)$. Thus, each voter's preference for $B$ over $A$ is treated equally. A voting method satisfies Tiebreaking Compensation if for any preference profile $\mathbf{P}$ and set $\mathsf{Y}$ of candidates in $\mathbf{P}$, if two voters $i$ and $j$ each rank all candidates in $\mathsf{Y}$ in a big tie (in the political context, this would mean $i$ and $j$ both leave the candidates from  $\mathsf{Y}$ unranked), and $\mathbf{P}'$ is obtained from $\mathbf{P}$ just by $i$ and $j$ breaking that tie in their rankings \textit{in exactly opposite orders} (e.g., $i$ decides to rank the candidates in $\mathsf{Y}$ in the order $Y_1,Y_2,\dots,Y_n$, while $j$ decides to rank them in the order $Y_n,Y_{n-1},\dots,Y_1$), then $\mathcal{F}(\mathbf{P})=\mathcal{F}(\mathbf{P}')$. Thus, two voters resolving a tie in exactly opposite orders cancel out. In a similar spirit, a voting method satisfies Neutral Reversal (\citealt{Saari2003}) if adding to a profile two voters whose rankings are the exact reverses of each other (and contain no ties) does not change the result. Ding et al.~\citeyearpar{Ding2025} prove that satisfying these axioms is equivalent to being margin-based for voting methods on the domain of preference profiles relevant to political elections.
\begin{theorem}[\citealt{Ding2025}] If $\mathcal{F}$ is a voting method on the domain of preference profiles of linear orders with bottom indifference,\footnote{To be precise, where $\mathcal{V}$ and $\mathcal{C}$ are infinite sets of voters and candidates, respectively, consider all preference profiles for any nonempty finite $\mathsf{V}\subseteq\mathcal{V}$ and nonempty finite $\mathsf{C}\subseteq\mathcal{C}$ in which each voter's ranking is a linear order with bottom indifference.} then the following are equivalent:
\begin{enumerate}
\item $\mathcal{F}$ is margin-based;
\item $\mathcal{F}$ satisfies Preferential Equality, Tiebreaking Compensation, and Neutral Reversal.
\end{enumerate}
\end{theorem}
\noindent The invocation of this result is not meant to decisively settle that only margin-based voting methods should be used. But it does help clarify the normative content of what is in effect the first axiom in our axiomatic characterization of the proposed voting method from \citealt{Holliday2025}: the voting method is margin-based.

\subsection{Weighted tournament solutions}
 
With each margin-based voting method we can associate what is called a \textit{weighted tournament solution} (see \citealt{Fischer2016} for a survey). A \textit{tournament} $\mathrm{T}=(\mathsf{C},\to)$ is a pair of a nonempty set $\mathsf{C}$, which we take to be the set of candidates, and a binary relation $\to$ on $\mathsf{C}$, which we take to be the relation of \textit{head-to-head defeat}, satisfying the following properties:
\begin{enumerate}
\item asymmetry: for all $A,B$ in $\mathsf{C}$, it is not the case that $A\to B$ and $B\to A$.
\item connectedness: for all $A,B$ in $\mathsf{C}$, if $A\neq B$, then either $A\to B$ or $B\to A$.
\end{enumerate}
Connectedness rules out head-to-head ties. Such ties are possible but extremely unlikely in large political elections. In this paper, we will ignore the possibility of head-to-head ties. Thus, our axiomatization results will force a voting method to agree with the method defined in Section~\ref{Intro} only in cases without such ties; but we prefer not to get bogged down with the subtleties of tie-breaking in this paper.

If $\mathbf{P}$ is a preference profile with no zero margins between distinct candidates, then $\mathbf{P}$ gives rise to a tournament $(\mathsf{C},\to)$, which we may call $\mathbf{P}$'s \textit{associated tournament}, where $\mathsf{C}$ is the set of candidates from $\mathbf{P}$ and $A\to B$ if the margin of $A$ vs.~$B$ in $\mathbf{P}$ is positive. Conversely, McGarvey \citeyearpar{McGarvey1953} proved that for every tournament $\mathrm{T}$, there is a preference profile (of linear orders) whose associated tournament is $\mathrm{T}$.

However, it is not enough for us to know who defeats whom head-to-head; we also need to know \textit{by how much}. A \textit{weighted tournament} in the sense of \citealt{DeDonder2000} (there called a \textit{$0$-weighted tournament}) is a pair $T=(\mathsf{C},m)$ where $\mathsf{C}$ is a nonempty set of candidates, as above, but now $m$ is a function that assigns to each pair $A,B$ of candidates from $\mathsf{C}$ an integer $m(A,B)$ such that $m(A,B)+m(B,A)=0$, or equivalently, $m(A,B)=-m(B,A)$. Given any preference profile $\mathbf{P}$, we obtain a weighted tournament $(\mathsf{C},m)$, which we may call $\mathbf{P}$'s \textit{associated weighted tournament}, where $\mathsf{C}$ is the set of candidates from $\mathbf{P}$ and for any $A,B$ in $\mathsf{C}$,  $m(A,B)$ is the margin of $A$ vs.~$B$ in $\mathbf{P}$. Conversely, Debord \citeyearpar{Debord1987} proved that for every weighted tournament $T$, there is a preference profile whose associated weighted tournament is $T$.\footnote{If we want $\mathbf{P}$ to be a profile of linear orders, then $T$ must be such that all numbers in $\{m(A,B)\mid A,B\in\mathsf{C}\mbox{ and } A\neq B\}$ have the same parity.}

By a \textit{uniquely-weighted} tournament, we mean a weighted tournament $(\mathsf{C},m)$ in which no two head-to-head matches between distinct candidates are decided by exactly the same margin, i.e., if $A\neq B$, $C\neq D$, and $(A,B)\neq (C,D)$, then $m(A,B)\neq m(C,D)$. It follows that if $A\neq B$, then $m(A,B)\neq 0$ (since otherwise $m(A,B)=m(B,A)$), so a uniquely-weighted tournament can be associated with a tournament $(\mathsf{C},\to)$ where $A\to B$ if $m(A,B)>0$. Almost every large-scale political election results in a uniquely-weighted tournament. Since, as noted above, we do not wish to get bogged down in the subtleties of tie-breaking in this paper, our main results will concern uniquely-weighted tournaments only.

A \textit{weighted tournament solution} is a function whose domain is some set of weighted tournaments such that for any $T=(\mathsf{C},m)$ in the domain, $F(T)$ is a nonempty subset of $\mathsf{C}$. Given a margin-based voting method $\mathcal{F}$, we obtain a weighted tournament solution $F$ whose domain consists of the associated weighted tournaments of the profiles in the domain of $\mathcal{F}$: for any $T$ in the domain, we set $F(T)=\mathcal{F}(\mathbf{P})$, where $\mathbf{P}$ is one of the profiles whose associated weighted tournament is $T$; since $\mathcal{F}$ is margin-based, any choice of such a profile produces the same result, so $F$ is well defined.

\subsection{Four weighted tournament solutions for Final Four elections}\label{FourSolutions}

From now until Section~\ref{Five}, we will focus on weighted tournament solutions defined on the set of uniquely-weighted tournaments with up to four candidates, which can be used for Final Four elections. As examples of such solutions, we will discuss the following:

\begin{itemize}
\item Copeland (\citealt{Copeland1951}): select the candidates with the most head-to-head wins.\footnote{\label{TiesCopeland}When head-to-head ties are allowed, Copeland is usually defined so that we score each candidate by their number of wins \textit{minus their number of losses} and then select the candidates with the highest score. This selects the same candidates as the following scheme: give each candidate 1 point for a win, $1/2$ a point for a tie, and $0$ points for a loss, and then select the candidates with the most points. However, there are other ways of handling ties, such as the Llull method that assigns a candidate $1$ point for a win or tie and $0$ points for a loss (see \citealt{Faliszewski2009}).} Formally, \[\mathrm{Copeland}(\mathsf{C},m)=\underset{X\in\mathsf{C}}{\mathrm{argmax}} \;| \{Y\in \mathsf{C}\mid m(X,Y)>0 \} |,\]
where $|S|$ is the cardinality of the set $S$.
\item Minimax (\citealt{Simpson1969}, \citealt{Kramer1977}): select the candidate whose worst head-to-head loss is smallest. Formally,
\[\mathrm{Minimax}(\mathsf{C},m)= \underset{X\in\mathsf{C}}{\mathrm{argmin}}\; \mathrm{max}(\{m(Y,X)\mid Y\in\mathsf{C}\mbox{ and }m(Y,X)>0\}),\]
where we set $\mathrm{max}(\varnothing)=0$.
\item Proposed method (\citealt{Holliday2025}): select the candidate with the most head-to-head wins; if multiple candidates tie for the most wins, then select the one who has the smallest head-to-head loss. Formally,
\[\mathrm{Proposed}(\mathsf{C},m)=  \underset{X\in \mathrm{Copeland}(\mathsf{C},m)}{\mathrm{argmin}}\; \iota (\{ m(Y,X)\mid Y\in\mathsf{C}\mbox{ and } m(Y,X) > 0 \}),
 \]
 where $\iota(\varnothing)=0$ and $\iota(\{n\})=n$, taking advantage of the fact noted in Section~\ref{Intro} (Footnote~\ref{AtMostOne}) that in a tournament with four candidates, there is a candidate with at most one loss, so each Copeland winner has at most one loss, so the set to which $\iota$ is applied has at most one element.
\item Variant method (\citealt[Notes 11 and 22]{Holliday2025}): select the candidate with the most head-to-head wins; if multiple candidates tie for the most wins, then among the candidates in that group, select the candidate who beats each of the others in that group head-to-head; if there is no such candidate, then select the candidate in that group with the smallest head-to-head loss to another candidate in that group. Formally,
\[\mathrm{Variant}(\mathsf{C},m)=  \underset{X\in \mathrm{Copeland}(\mathsf{C},m)}{\mathrm{argmin}}\; \iota (\{ m(Y,X)\mid Y\in \mathrm{Copeland}(\mathsf{C},m)\mbox{ and } m(Y,X) > 0 \}),
 \]
 with $\iota$ defined as above.
\end{itemize}
\noindent In the next section, we will discuss an example in which the proposed method and variant method give different results. Recall that when there are five or more candidates, it is no longer guaranteed that each Copeland winner has at most one loss (see Figure~\ref{Penta}). Then we face the question of whether we should choose the Copeland winner whose worst loss is smallest (replacing $\iota$ with $\mathrm{max}$), as suggested in \citealt[\S~2.3]{Holliday2025}, or the Copeland winner whose smallest loss is smallest (replacing $\iota$ with $\mathrm{min}$), as we will suggest in Section~\ref{Five}.

\begin{figure}[h!]
\begin{center}
\begin{tikzpicture} 

\node[circle,draw, minimum width=0.3in] at (0,0) (W) {$W$}; 
\node[circle,draw,minimum width=0.3in] at (4,0) (E) {$E$}; 
\node[circle,draw,minimum width=0.3in, fill=medgreen!75] at (2,2) (N) {$N$}; 
\node[circle,draw,minimum width=0.3in] at (2,-2) (S) {$S$}; 

\path[->,draw,thick] (N) to node { } (W);
\path[->,draw,thick] (N) to node   { } (E);
\path[->,draw,thick] (N) to node    { } (S);
\path[->,draw,thick] (W) to node {} (S);
\path[->,draw,thick] (W) to node {} (E);
\path[->,draw,thick] (E) to node {} (S);

\node at (2,-2.75)  {Condorcet winner}; 
\node at (2,-3.25)  {linear order};

\end{tikzpicture}\qquad\qquad \begin{tikzpicture} 

\node[circle,draw, minimum width=0.3in] at (0,0) (W) {$W$}; 
\node[circle,draw,minimum width=0.3in] at (4,0) (E) {$E$}; 
\node[circle,draw,minimum width=0.3in, fill=medgreen!75] at (2,2) (N) {$N$}; 
\node[circle,draw,minimum width=0.3in] at (2,-2) (S) {$S$}; 

\path[->,draw,thick] (N) to node { } (W);
\path[->,draw,thick] (N) to node   { } (E);
\path[->,draw,thick] (N) to node    { } (S);
\path[->,draw,very thick] (S) to node {} (W);
\path[->,draw,very thick] (W) to node {} (E);
\path[->,draw,very thick] (E) to node {} (S);

\node at (2,-2.75) {Condorcet winner}; 
\node at (2,-3.25) {bottom cycle};

\end{tikzpicture}\vspace{.2in}

\begin{tikzpicture} 

\node[circle,draw, minimum width=0.3in, fill=gray!75] at (0,0) (W) {$W$}; 
\node[circle,draw,minimum width=0.3in, fill=gray!75] at (4,0) (E) {$E$}; 
\node[circle,draw,minimum width=0.3in, fill=medgreen!75] at (2,2) (N) {$N$}; 
\node[circle,draw,minimum width=0.3in] at (2,-2) (S) {$S$}; 

\path[->,draw,very thick] (W) to node[fill=white] {\footnotesize small} (N);
\path[->,draw,very thick] (N) to node[fill=white]   {\footnotesize medium} (E);
\path[->,draw,very thick] (E) to[pos=.75]  node[fill=white]   {\footnotesize large} (W);
\path[->,draw,thick] (W) to node {} (S);
\path[->,draw,thick] (N) to node {} (S);
\path[->,draw,thick] (E) to node {} (S);

\node at (2,-2.75) {ascending top cycle}; 
\node at (2,-3.25)  {Condorcet loser};

\end{tikzpicture}\qquad\qquad \begin{tikzpicture} 

\node[circle,draw, minimum width=0.3in, fill=gray!75] at (0,0) (W) {$W$}; 
\node[circle,draw,minimum width=0.3in, fill=gray!75] at (4,0) (E) {$E$}; 
\node[circle,draw,minimum width=0.3in, fill=medgreen!75] at (2,2) (N) {$N$}; 
\node[circle,draw,minimum width=0.3in] at (2,-2) (S) {$S$}; 

\path[->,draw,very thick] (W) to node[fill=white] {\footnotesize small} (N);
\path[->,draw,very thick] (N) to node[fill=white]   {\footnotesize large} (E);
\path[->,draw,very thick] (E) to[pos=.735]  node[fill=white]   {\footnotesize medium} (W);
\path[->,draw,thick] (W) to node {} (S);
\path[->,draw,thick] (N) to node {} (S);
\path[->,draw,thick] (E) to node {} (S);

\node at (2,-2.75) {descending top cycle}; 
\node at (2,-3.25)  {Condorcet loser};

\end{tikzpicture} \vspace{.2in}

\begin{tikzpicture} 

\node[circle,draw, minimum width=0.3in] at (0,0) (W) {$W$}; 
\node[circle,draw,minimum width=0.3in, fill=gray!75] at (4,0) (E) {$E$}; 
\node[circle,draw,minimum width=0.3in, fill=medgreen!75] at (2,2) (N) {$N$}; 
\node[circle,draw,minimum width=0.3in] at (2,-2) (S) {$S$}; 

\path[->,draw,very thick] (W) to node[fill=white] {\footnotesize smaller} (N);
\path[->,draw,very thick] (N) to node[fill=white]   {\footnotesize larger} (E);
\path[->,draw,thick] (E) to[pos=.75]  node   {} (W);
\path[->,draw,very thick] (S) to node {} (W);
\path[->,draw,thick] (N) to node {} (S);
\path[->,draw,very thick] (E) to node {} (S);

\node at (2,-3) {SL four cycle};

\end{tikzpicture}\qquad\qquad \begin{tikzpicture} 

\node[circle,draw, minimum width=0.3in] at (0,0) (W) {$W$}; 
\node[circle,draw,minimum width=0.3in, fill=medgreen!75] at (4,0) (E) {$E$}; 
\node[circle,draw,minimum width=0.3in, fill=gray!75] at (2,2) (N) {$N$}; 
\node[circle,draw,minimum width=0.3in] at (2,-2) (S) {$S$}; 

\path[->,draw,very thick] (W) to node[fill=white] {\footnotesize larger} (N);
\path[->,draw,very thick] (N) to node[fill=white]   {\footnotesize smaller} (E);
\path[->,draw,thick] (E) to[pos=.75]  node   {} (W);
\path[->,draw,very thick] (S) to node {} (W);
\path[->,draw,thick] (N) to node {} (S);
\path[->,draw,very thick] (E) to node {} (S);

\node at (2,-3) {LS four cycle};

\end{tikzpicture}

\end{center}
\caption{Classes of weighted tournaments for four candidates from \citealt{Holliday2025}. In each case, the winning candidate according to the proposed method is shown in green.}\label{ElectionTypes}
\end{figure}
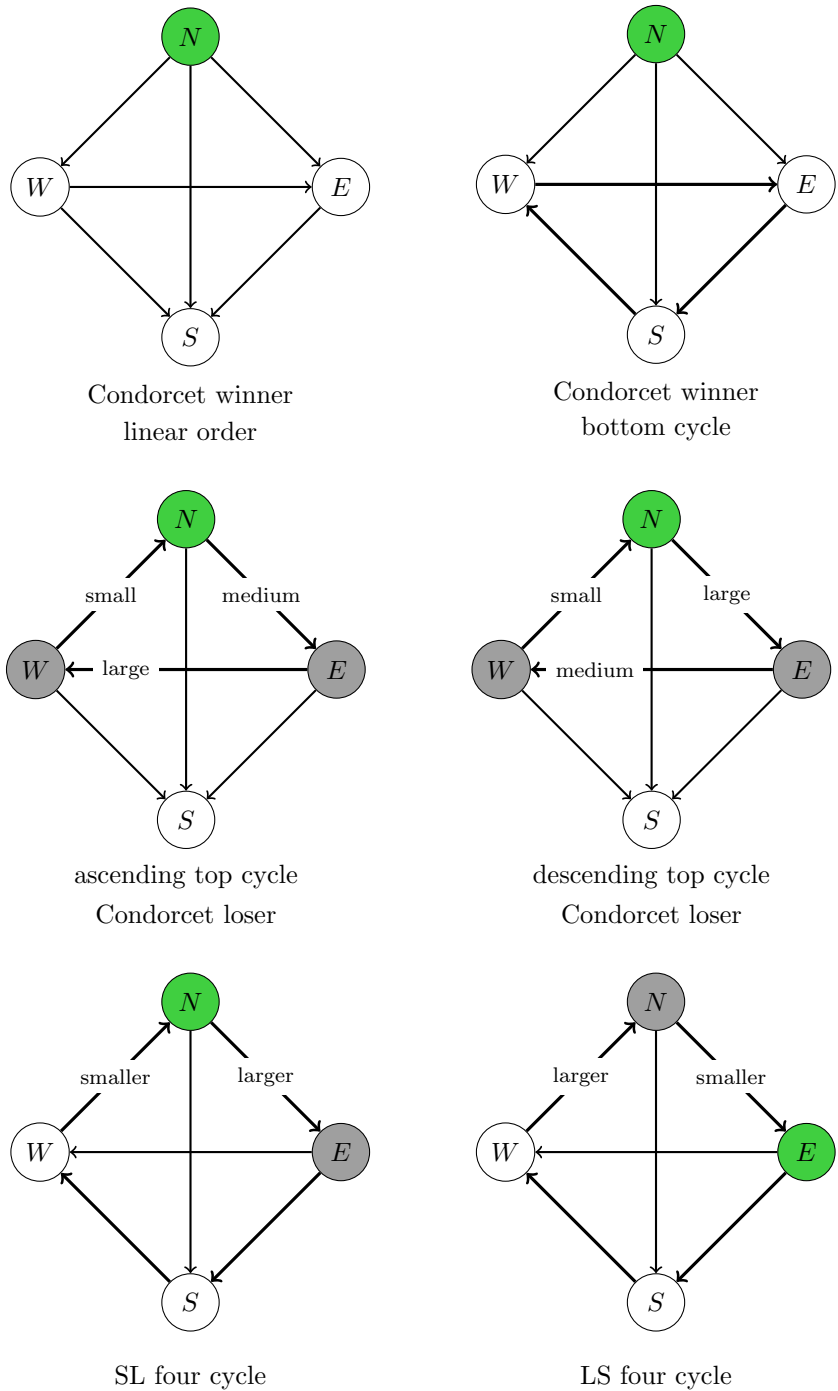

As shown in \citealt{Holliday2025}, every uniquely-weighted tournament for four candidates falls into one of the six types in Figure \ref{ElectionTypes}, where we give the ordinal comparisons of margin sizes that are relevant to the choice of winner by the proposed method. The winner according to the proposed method is shown in green; if there are additional Copeland winners, they are shown in gray.

\section{Axioms}\label{Axioms}

In this section, we introduce the axioms used in our characterization of the proposed method for Final Four elections, dedicating one subsection to each axiom. Table \ref{AxTable} shows the patterns of satisfaction and violation of these axioms, to be explained below, by the four voting methods from Section~\ref{Prelim}.

\begin{table}[h]
\begin{center}
\begin{tabular}{l|c|c|c|c}
Axiom & Copeland & Minimax & Proposed method & Variant method \\
\hline
Proximity to Condorcet  & $\checkmark$ & $\checkmark$ & $\checkmark$ & --  \\
Independence of Irrelevant Defeats & $\checkmark$ & $\checkmark$ & $\checkmark$ & --  \\
Win Monotonicity & $\checkmark$ & $\checkmark$ & $\checkmark$ & $\checkmark$    \\
Win Dominance & $\checkmark$ & -- & $\checkmark$& $\checkmark$ \\
Rare Ties & -- & $\checkmark$ & $\checkmark$& $\checkmark$  \\

\end{tabular}
\end{center}
\caption{Satisfaction ($\checkmark$) and violation (--) of axioms by voting methods.}\label{AxTable}
\end{table}

\subsection{Proximity to Condorcet}

Suppose one takes the Condorcet winner to be the gold standard for a clear winner in an election (we will derive the Condorcet Criterion from other axioms in Proposition~\ref{CondorcetDerived} below). Then as suggested in Section~\ref{Intro}, the question of what to do when there is no Condorcet winner in an election seems to have a very natural answer: elect the candidate who is \textit{closest} to being a Condorcet winner. The difficulty is that there are multiple plausible definitions of closeness, which lead to different results, as studied in the literature on \textit{distance rationalizations} of voting rules (see \citealt{Elkind2016} for an overview). For example, perhaps the distance of a candidate $C$ from being the Condorcet winner is to be measured in terms of the \textit{number of new voters} one would need to add to the election in order to make $C$ the Condorcet winner. Then the candidate closest to being the Condorcet winner is the Minimax winner, as observed by Young~\citeyearpar{Young1977}.  Or perhaps the distance of $C$ from being the Condorcet winner is to be measured in terms of the number of head-to-head results that would have to flip (from $X$ beating $Y$ to $Y$ beating $X$) in order for $C$ to be the Condorcet winner.  Then the candidate closest to being the Condorcet winner is the Copeland winner, as observed by Klamler \citeyearpar{Klamler2005}. Or perhaps distance is to be measured in terms of some combination of these factors or other factors.\footnote{As an example of another factor, Young \citeyearpar{Young1977} considers the number of voters who would have to be \textit{removed} from a preference profile in order to make $C$ the Condorcet winner. Unlike the definitions of closeness in the main text, this definition requires more information about a preference profile than the head-to-head margins, since it depends on how many existing voters have each type of ranking.} The problem is that what counts as the closest counterfactual situation in which $C$ would have been the Condorcet winner is not obvious. This is related to the more general problem, long investigated by philosophers, of determining what counts as the closest counterfactual situation in which something different happened than what actually happened (see, e.g.,~\citealt{Lewis1979}).

Our approach is to assume a relatively uncontroversial sufficient condition for one candidate to count as closer than another to being the Condorcet winner, at least in the context of weighted tournaments:
\begin{itemize}
\item \textbf{Proximity to Condorcet}: for any weighted tournament $T=(\mathsf{C},m)$, candidates $A,B$ in $\mathsf{C}$, if $A$ can be made a Condorcet winner by improving one of $A$'s margins against another candidate by some non-negative integer $n$,\footnote{That is, for some candidate $X\neq A$ in $\mathsf{C}$, $A$ is the Condorcet winner in the weighted tournament $T'=(\mathsf{C},m')$ where $m'(A,X)=m(A,X)+n$ and hence $m'(X,A)=m(X,A)-n$, while $m'$ is otherwise the same as $m$.} while $B$ cannot be made a Condorcet winner by improving \textit{every one} of $B$'s margins by the same $n$,\footnote{That is, $B$ is not the Condorcet winner in the weighted tournament $T^\#=(\mathsf{C},m^\#)$ where for \textit{every} candidate $V\neq B$ in $\mathsf{C}$, $m^\#(B,V)=m(B,V)+n$ and hence $m^\#(V,B)=m(V,B)-n$, while $m^\#$ is otherwise the same as $m$.} then $F(T)\neq \{B\}$. 
\end{itemize}
This axiom is consistent with both the Minimax and Copeland notions of closeness discussed above: the assumptions on $A$ and $B$ imply that fewer voters must be added to  the election to make $A$ a Condorcet winner than to make $B$ a Condorcet winner, and at most one head-to-head result must be flipped in order to make $A$ the Condorcet winner, while at least one must be flipped in order to make $B$ the Condorcet winner. Hence $A$ is at least as close to being the Condorcet winner as $B$ is according to both notions of closeness, and $A$ is strictly closer according to the notion based on voter addition; so $B$ should not be the unique winner.

Thus, both Minimax and Copeland satisfy Proximity to Condorcet in our sense. The proposed method also satisfies the axiom, since if $A$ has at most one loss, and $B$ has a larger loss than $A$ does, then the proposed method will not elect $B$. However, we will show in the next subsection that the variant method violates Proximity to Condorcet, at the same time that we show it violates another axiom.

Proximity to Condorcet sidesteps difficult cases like the following: in $T$, $A$ only has one loss, by a margin of $6$, whereas $B$ has two losses,  by margins of $4$ and $2$. In this case, in which Copeland favors $A$ whereas Minimax favors $B$, the axiom simply does not apply. There is no $n$ such that $A$ can be made a Condorcet winner in a $T_A$ obtained from $T$ by improving one of $A$'s margins by $n$ while $B$ cannot be made a Condorcet winner in a $T_B$ obtained from $T$ by improving \textit{every one of} $B$'s margins by $n$. We would have to improve one of $A$'s margins by $n\geq 7$ in order to make $A$ the Condorcet winner in $T_A$, but if we were to improve every one of $B$'s margins by this $n$, that would make $B$ the Condorcet winner in $T_B$. On the other hand, there is no $n$ such that $B$ can be made a Condorcet winner by improving one of $B$'s margins by $n$, since $B$ has two losses. 

Thus, there are many cases in which Proximity to Condorcet does not apply.  In these cases, our other axioms, introduced below, will jointly determine who wins.

\subsection{Independence of Irrelevant Defeats}

Arrow's \citeyearpar{Arrow1963} famous axiom of the Independence of Irrelevant Alternatives (IIA) states that to determine the \textit{social ranking} of two alternatives, $A$ and $B$, all that matters is how voters rank $A$ vs.~$B$; that is, if we consider two preference profiles such that for each voter $i$, $i$ ranks $A$ vs.~$B$ in the same way in the two profiles, then the social ranking of $A$ vs.~$B$ must be the same for both profiles. The moral of Arrow's Impossibility Theorem and many later generalizations is plausibly that IIA is too strong (see \citealt{HP2020,HP2021}, \citealt{HK2025}, and references therein); in any case, none of the ranked voting methods used in practice satisfy it. By contrast, in his axiomatization of the Copeland ranking rule, Rubinstein \citeyearpar{Rubinstein1980} identifies a weakening of IIA that is satisfied not only by the Copeland ranking but also by other ranking rules. According to this weakened version, when determining the social ranking of two alternatives, $A$ and $B$, it does not matter how voters rank some $C$ vs.~$D$, where $C$ and $D$ are distinct from both $A$ and $B$; but it may matter how $A$ performs against $C$/$D$ and $B$ performs against $C$/$D$. This is quite intuitive: how two other candidates perform against each other is irrelevant to the social ranking of $A$ vs.~$B$, but how $A$ and $B$ perform against other candidates is not irrelevant to the social ranking of $A$ vs.~$B$. What are \textit{irrelevant} to $A$ vs.~$B$ are not \textit{alternatives} like $C$ but rather certain \textit{defeats} like $C$'s defeat of $D$, whereas $A$'s defeat of $C$ may well be relevant. Thus, Rubinstein's principle may be called the Independence of Irrelevant Defeats.

In this paper, we are interested in picking a winning candidate instead of ranking the candidates, but Rubinstein's principle can be applied in our setting, as follows:
\begin{itemize}
\item \textbf{Independence of Irrelevant Defeats}: for any weighted tournament $T$ and distinct candidates $A,B$, if $F(T)=\{A\}$ and $T'$ is obtained from $T$ by changing a margin involving neither $A$ nor $B$, while keeping all other margins the same, then $F(T')\neq \{B\}$. 
\end{itemize}
The idea is that if $F(T)=\{A\}$, this tells us that $A$ is socially ranked above $B$. Then by Rubinstein's principle, if $T'$ is obtained by changing a margin that is \textit{irrelevant} to $A$ vs.~$B$, this should not cause $B$ to be socially ranked above $A$, so we should not get that $F(T')=\{B\}$. 

Clearly Copeland, Minimax, and the proposed method all satisfy Independence of Irrelevant Defeats, since changing the margin of $C$ vs.~$D$ does not affect $A$'s or $B$'s number of wins or sizes of losses. By contrast, the variant method violates Independence of Irrelevant Defeats. For example, on the left of Figure \ref{IIDViolation}, the candidates with the most wins are $N$ and $E$, each with two wins, but $N$ beats $E$ head-to-head, so the variant method selects $N$, despite $N$ having a larger loss than $E$; by contrast, on the right, the candidates with the most wins are $N$, $W$, and $E$, and none of these candidates beats both of the others head-to-head, so the variant method selects the one with the smallest loss, namely $E$. Since changing the irrelevant head-to-head result between $W$ and $S$ changes the selected candidate from $N$ to $E$, the variant method violates Independence of Irrelevant Defeats. By contrast, the proposed method selects $E$ in both cases, since among the candidates with two wins, $E$ has the smallest loss and is therefore closest to being the Condorcet winner. As promised in the previous subsection, this is also a violation of Proximity to Condorcet by the variant method, since on the left of Figure~\ref{IIDViolation}, the variant method selected $N$ even though $E$ is closer to being the Condorcet winner. The same example will be used in Section~\ref{Five} to show that the variant method violates Immunity to Spoilers.

\begin{figure}[h!]
\begin{center}
 \begin{minipage}{2.5in}
 \begin{center}\begin{tikzpicture} 

\node[circle,draw, minimum width=0.3in] at (0,0) (W) {$W$}; 
\node[circle,draw,minimum width=0.3in, fill=gray!75] at (4,0) (E) {$E$}; 
\node[circle,draw,minimum width=0.3in, fill=gray!75] at (2,2) (N) {$N$}; 
\node[circle,draw,minimum width=0.3in] at (2,-2) (S) {$S$}; 

\path[->,draw, thick] (W) to node[fill=white] {\footnotesize large} (N);
\path[->,draw, thick] (N) to node[fill=white]   {\footnotesize small} (E);
\path[->,draw, thick] (E) to[pos=.74]  node[fill=white]   {\footnotesize medium}  (W);
\path[->,draw, thick,red] (S) to node {} (W);
\path[->,draw, thick] (N) to node {} (S);
\path[->,draw, thick] (E) to node {} (S);

\node at (2,-3) {LS four cycle};

\end{tikzpicture}\end{center}\end{minipage}\begin{minipage}{2.5in}\begin{center}\begin{tikzpicture} 

\node[circle,draw, minimum width=0.3in, fill=gray!75] at (0,0) (W) {$W$}; 
\node[circle,draw,minimum width=0.3in, fill=gray!75] at (4,0) (E) {$E$}; 
\node[circle,draw,minimum width=0.3in, fill=gray!75] at (2,2) (N) {$N$}; 
\node[circle,draw,minimum width=0.3in] at (2,-2) (S) {$S$}; 

\path[->,draw, thick] (W) to node[fill=white] {\footnotesize large} (N);
\path[->,draw, thick] (N) to node[fill=white]   {\footnotesize small} (E);
\path[->,draw, thick] (E) to[pos=.75]  node[fill=white]   {\footnotesize medium}  (W);
\path[->,draw, thick, red] (W) to node {} (S);
\path[->,draw,thick] (N) to node {} (S);
\path[->,draw, thick] (E) to node {} (S);

\node at (2,-3) {top cycle};

\end{tikzpicture}
\end{center}
\end{minipage}
\end{center}
\caption{A violation of Independence of Irrelevant Defeats by the variant method: on the left, the variant method selects $N$, whereas on the right, the variant method selects $E$, while the only change is the irrelevant head-to-head result between $W$ and $S$. The proposed method selects $E$ in both cases.}\label{IIDViolation}
\end{figure}
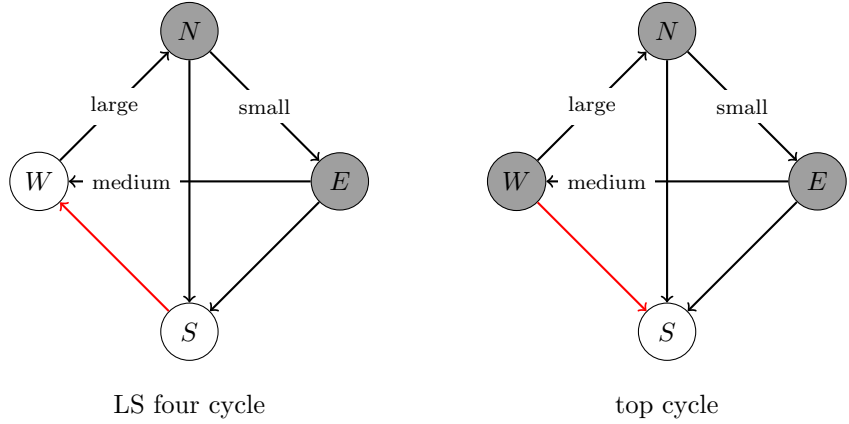

\subsection{Win Monotonicity}

The next axiom says that if we imagine changing an election by strengthening a win of the elected candidate $A$ by as much as we strengthen a win of an unelected candidate $B$ (against some candidate distinct from $A$), then $A$ should still be elected in the modified election: 

\begin{itemize}
\item \textbf{Win Monotonicity}: for any weighted tournament $T$, candidates $A,B$, and non-negative integer $n$, if $F(T)=\{A\}$, and $T'$ is obtained from $T$ by improving one of $A$'s margins of victory by $n$ and one of $B$'s margins of victory over some $X\neq A$ by $n$, then $F(T')=\{A\}$.\footnote{By `improving one of $A$'s margins of victory by $n$', we mean that for some $Y$ such that $m(A,Y)>0$ in $T=(\mathsf{C},m)$, we change $m$ to $m'$ such that $m'(A,Y)=m(A,Y)+n$ and hence $m'(Y,A)=m(Y,A)-n$, while $m'$ is otherwise the same as $m$, and similarly for improving $B$'s positive margin vs.~some $X\neq A$.}
\end{itemize}
Copeland, Minimax, the proposed method, and the variant method all satisfy Win Monotonicity, since the change from $T$ to $T'$ does not change any candidate's number of wins, and it neither changes the size of $A$'s losses nor decreases the size of $B$'s losses. Note that it would be too strong to replace `margins of victory' in the definition with simply `margins' (i.e., `margins of victory or loss'), since strengthening one of $A$'s existing majority victories may be less important than flipping one of $B$'s majority \textit{losses} to a majority \textit{victory}, even if the same change in margin is involved in both cases. For example, flipping one of $B$'s losses to a victory with a change in margin of $n$ may make $B$ the Condorcet winner, in which case $F(T')\neq\{A\}$ by Proximity to Condorcet. Thus, we restrict attention to head-to-head wins in the statement of Win Monotonicity.

\subsection{Win Dominance}\label{WinDom}

All the axioms introduced so far are satisfied by both Copeland and Minimax. It is our last two axioms that will distinguish between these methods. The first one descends from the \textit{covering relation} of Fishburn~\citeyearpar{Fishburn1977} and Miller \citeyearpar{Miller1980}. In a tournament, $A$ \textit{covers} $B$ if $A$ not only defeats $B$ but also defeats every candidate whom $B$ defeats.\footnote{In a \textit{weak} tournament allowing head-to-head ties between candidates, there are several non-equivalent versions of covering (see \citealt{Duggan2013}), but in the tournaments that we are interested in, all the versions become equivalent.} Fishburn and Miller propose choosing candidates from the \textit{uncovered set}, i.e., the set of candidates who are not covered by any other candidate. We will generalize this idea to weighted tournaments\footnote{\label{Supercover}A different generalization of the covering relation to weighted tournaments appears in \citealt{Dutta1999} (Definition 3.2) and \citealt{Fernandez2018}: say that $A$ $m$-covers $B$ if $m(A,B)>0$ and for all $X\in \mathsf{C}\setminus\{A,B\}$, $m(A,X)\geq m(B,X)$.} as follows: $A$ \textit{dominates $B$ in wins} if $A$ not only defeats $B$ but also defeats every candidate whom $B$ defeats \textit{by at least as much as $B$ does}.\footnote{That is, $m(A,B)>0$ and for all $X\in \mathsf{C}\setminus\{A,B\}$, if $m(B,X)>0$, then $m(A,X)\geq m(B,X)$. Note that this is a weaker condition than that in Footnote~\ref{Supercover}.} The next axiom simply says that a candidate who is dominated in wins cannot be the winner:
\begin{itemize}
\item \textbf{Win Dominance}: for any weighted tournament $T$ and candidates $A,B$, if $A$ dominates $B$ in wins, then $F(T)\neq \{B\}$.
\end{itemize}
This axiom is weaker than requiring that winners belong to the uncovered set, since $B$ may be covered by $A$ and yet not dominated in wins when we look at the margins of victory. Still, this is plausibly the least normatively compelling axiom in our list, since it rules out electing candidates who are dominated in wins even if those candidates do well from the perspective of having small losses.\footnote{Indeed, the axiom is stronger than requiring that winners not be $m$-covered as in Footnote \ref{Supercover}, since if $A$ $m$-covers $B$, then $A$ dominates $B$ in wins, but the converse implication does not generally hold.}

Copeland, the proposed method, and the variant method all clearly satisfy Win Dominance, since if $A$ dominates $B$ in wins, then $A$ has at least one more head-to-head win than $B$ does, so $B$ is not among the candidates with the most wins. By contrast, Minimax clearly violates it, since Minimax can select a Condorcet loser, who is dominated in wins by every other candidate.  As suggested above, this is hardly a decisive argument against Minimax, since perhaps it is more important to prioritize small losses than wins of any size. But our aim here is not to decisively argue against alternative methods like Minimax\footnote{As suggested in Section~\ref{Five}, perhaps the most important practical advantage of the proposed method over Minimax is that voters may find the proposed method easier to understand.} but rather to identify a package of principles that lead uniquely to the proposed voting method.

\subsection{Rare Ties}

Our final axiom highlights the problem with Copeland as a weighted tournament solution: it leads to too many ties. In the bottom four types of weighted tournaments in Figure~\ref{ElectionTypes}, Copeland fails to select a unique winner. Of course, if the margins of defeat are tied, e.g., if we have a top cycle in which all candidates in the top cycle beat each other by the same margin, then we may be forced to declare a tie among candidates. But our final axiom says that if there are no ties between margins of defeat, so the tournament is \textit{uniquely weighted}, then there should be no ties between candidates either:
\begin{itemize}
\item \textbf{Rare Ties}: for any uniquely-weighted tournament $T$,  $F(T)$ contains only one candidate.\footnote{The analogous axiom for preference profiles instead of weighted tournaments is called `quasi-resoluteness' in \citealt{HollidayEtAl2024}.}
\end{itemize}
The motivation for this axiom is pragmatic: for large-scale political elections, we should try to make ties extremely rare, and if all head-to-head margins are distinct, then we are not forced to declare a tie (as the many voting methods satisfying Rare Ties show), so we should not declare a tie.

Minimax, the proposed method, and the variant method all satisfy this axiom, since if all margins are distinct, then there cannot be a tie when comparing the sizes of two candidates' losses. Other weighted tournament solutions satisfying the axiom include Ranked Pairs (\citealt{Tideman1987}), Beat Path (\citealt{Schulze2011}), River (\citealt{Doring2025}), and Stable Voting (\citealt{HP2023b}).

\begin{remark}\label{BordaNote}Not all standard weighted tournament solutions satisfy Rare Ties, but they can be turned into ones that do. For example, consider the Copeland-Global-Borda solution\footnote{The reason for the `Global' in the name of this method will be explained in Remark~\ref{LocalGlobal}.} (named in part after \citealt{Borda1781}) on weighted tournaments (see \citealt[p.~77]{Dummett1997}, \citealt{MaskinDasgupta2004}, and \citealt{Foley2021} for Copeland-Global-Borda defined on preference profiles), which selects from among the Copeland winners the ones with the greatest \textit{symmetric Borda score}, where the symmetric Borda score of a candidate is the sum of their margins against all candidates (\citealt[p.~28]{Zwicker2016}). Formally,
\[\mathrm{CGB}(\mathsf{C},m)= \underset{X\in\mathrm{Copeland}(\mathsf{C},m)}{\mathrm{argmax}} \;\underset{Y\in\mathsf{C}}\sum m(X,Y).\]
The reason this method violates Rare Ties is that even in a uniquely-weighted tournament, Copeland winners may tie in the sum of their margins. For example, one Copeland winner may have margins of $13$, $3$, and $-1$, which sum to $15$, while another has margins of $11$, $7$, and $-3$, which also sum to $15$. Of course, such tied sums will be extraordinarily rare in large political elections, so satisfying Rare Ties as stated above is sufficient but not necessary for making ties rare in practice.

For any weighted tournament solution $F$  that does not satisfy Rare Ties in our sense, one can modify it to a solution that satisfies Rare Ties by adding a final Minimax tie-breaking step:
\[F^+(\mathsf{C},m)= \underset{X\in F(\mathsf{C},m)}{\mathrm{argmin}} \,\mathrm{max}(\{m(Y,X)\mid Y\in\mathsf{C}\mbox{ and }m(Y,X)>0\}).\]
Both CGB and CGB$^+$ satisfy our axioms of Independence of Irrelevant Defeats, Win Monotonicity, and Win Dominance, as we will argue in the proof of Proposition \ref{IndProp}. However, they both violate Proximity to Condorcet, as we will also show in the proof of Proposition \ref{IndProp}. These Borda-based methods also violate Immunity to Spoilers, as we will show in Section \ref{Five}. From a practical perspective, they are also less simple conceptually than the proposed method, insofar as one has to think about checking, for each Copeland winner, the sum of all their margins, both positive and negative, whereas with the proposed method, one simply thinks about checking their smallest loss.\end{remark}

\section{Proof of characterization}\label{Proof}

Let us now restate our axioms all in one place, leaving the quantification over all weighted tournaments~$T$, distinct candidates $A,B$, and non-negative integer $n$ implicit:
\begin{itemize}
\item \textbf{Proximity to Condorcet}: if $A$ can be made a Condorcet winner by increasing one of $A$'s margins by $n$, while $B$ cannot be made a Condorcet winner by increasing \textit{every one} of $B$'s margins by $n$, then $F(T)\neq \{B\}$.
\item \textbf{Independence of Irrelevant Defeats}: if $F(T)=\{A\}$ and $T'$ is obtained from $T$ by changing a margin involving neither $A$ nor $B$, then $F(T')\neq \{B\}$.
\item \textbf{Win Monotonicity}: if $F(T)=\{A\}$, and $T'$ is obtained from $T$ by improving one of $A$'s margins of victory by $n$ and one of $B$'s margins of victory over some $X\neq A$ by $n$, then $F(T')=\{A\}$.
\item \textbf{Win Dominance}: if $A$ dominates $B$ in wins, then $F(T)\neq \{B\}$.
\item \textbf{Rare Ties}: if $T$ is uniquely-weighted, then $F(T)$ contains only one candidate.
\end{itemize}

We are now in a position to prove our first main result.

\begin{theorem}\label{MainFour} Let $F$ be a weighted tournament solution for up to four candidates satisfying the axioms above. Then in any uniquely-weighted tournament, $F$ selects the same winner as the proposed method.
\end{theorem}

\begin{proof} Suppose $T$ is a uniquely-weighted tournament, so by \textbf{Rare Ties}, there is a unique winner in $T$ according to $F$. If $T$ has a Condorcet winner, then \textbf{Proximity to Condorcet} implies that $F$ must select the Condorcet winner uniquely, in agreement with the proposed method; trivially, the Condorcet winner can be made a Condorcet winner by improving any one of their margins by $0$, whereas no other candidate can be made a Condorcet winner by increasing every one of their margins by $0$, so no other candidate is the unique winner, so the Condorcet winner must be the unique winner. \textbf{Proximity to Condorcet} also implies that if $T$ has only three candidates in a majority cycle, then $F$ must select the candidate whose loss is smallest. 

Suppose $T$ has four candidates. Every uniquely-weighted tournament for four candidates falls into one of the classes shown in Figure~\ref{ElectionTypes}. As above, \textbf{Proximity to Condorcet} implies that $F$ must choose the Condorcet winner $N$ in the linear order and the bottom cycle, in agreement with the proposed method.

In the top cycles, \textbf{Win Dominance} implies that $S$ cannot win according to $F$, since $S$ is dominated in wins by all other candidates. By \textbf{Proximity to Condorcet}, neither $W$ nor $E$ can be the winner, since $N$'s one loss is smaller than $W$'s and $E$'s. Hence $N$ is the winner, in agreement with the proposed method.

Next, we turn to the four cycles. First, we argue that $W$ cannot win in any four cycle:

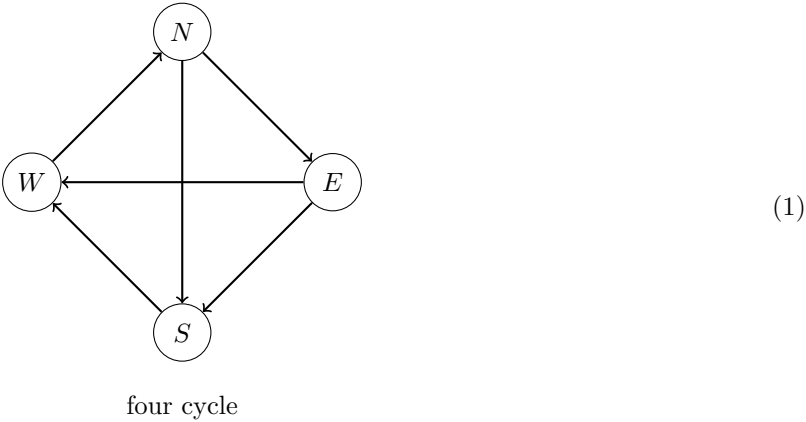
\begin{equation}
\begin{tikzpicture}[baseline={(current bounding box.center)}]

\node[circle,draw, minimum width=0.3in] at (0,0) (W) {$W$}; 
\node[circle,draw,minimum width=0.3in ] at (4,0) (E) {$E$}; 
\node[circle,draw,minimum width=0.3in ] at (2,2) (N) {$N$}; 
\node[circle,draw,minimum width=0.3in] at (2,-2) (S) {$S$}; 

\path[->,draw,  thick] (W) to node {} (N);
\path[->,draw,  thick] (N) to node  {}     (E);
\path[->,draw,thick] (E) to[pos=.75]  node   {} (W);
\path[->,draw,  thick] (S) to node {} (W);
\path[->,draw,thick] (N) to node {} (S);
\path[->,draw,  thick] (E) to node {} (S);

\node at (2,-3) {four cycle}; 
\end{tikzpicture}\label{FC}
\end{equation}
Suppose for contradiction that $W$ is the winner in (\ref{FC}). Then we can simultaneously strengthen $W$'s margin vs.~$N$ by some sufficiently large $n$ and $E$'s margin vs.~$S$ by $n$ to obtain a weighted tournament of the following form, while maintaining the property of being uniquely weighted:
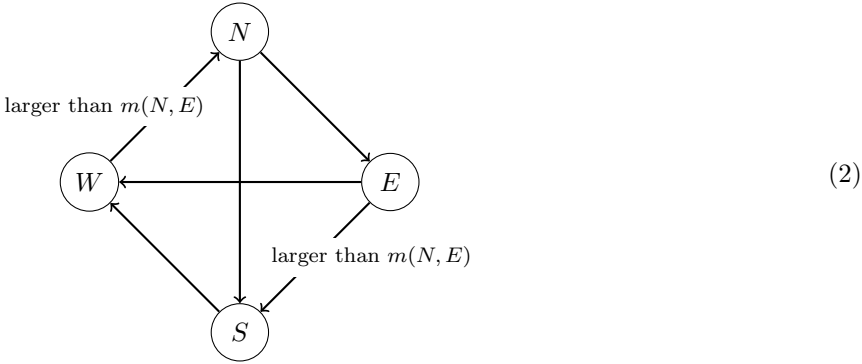
\begin{equation}
\begin{tikzpicture}[baseline={(current bounding box.center)}] 

\node[circle,draw, minimum width=0.3in] at (0,0) (W) {$W$}; 
\node[circle,draw,minimum width=0.3in ] at (4,0) (E) {$E$}; 
\node[circle,draw,minimum width=0.3in ] at (2,2) (N) {$N$}; 
\node[circle,draw,minimum width=0.3in] at (2,-2) (S) {$S$}; 

\path[->,draw,  thick] (N) to node   { } (E);
\path[->,draw,thick] (E) to[pos=.75]  node   {} (W);
\path[->,draw,  thick] (S) to node {} (W);
\path[->,draw,thick] (N) to node {} (S);
\path[->,draw,  thick] (E) to node {} (S);
\path[->,draw,  thick] (W) to node {} (N);

\node[fill=white] at (.2,1) {\footnotesize larger than $m(N,E)$}; 

\node at (-.8,0) {};
\node[fill=white] at (3.75,-1) {\footnotesize larger than $m(N,E)$}; 
\end{tikzpicture}\label{BoostEdges}
\end{equation}
By \textbf{Win Monotonicity}, $W$ must still be the winner in (\ref{BoostEdges}). Now suppose we flip $N$'s victory over $S$ to a victory for $S$ over $N$ with the largest of all margins (so we maintain the uniquely-weighted property):
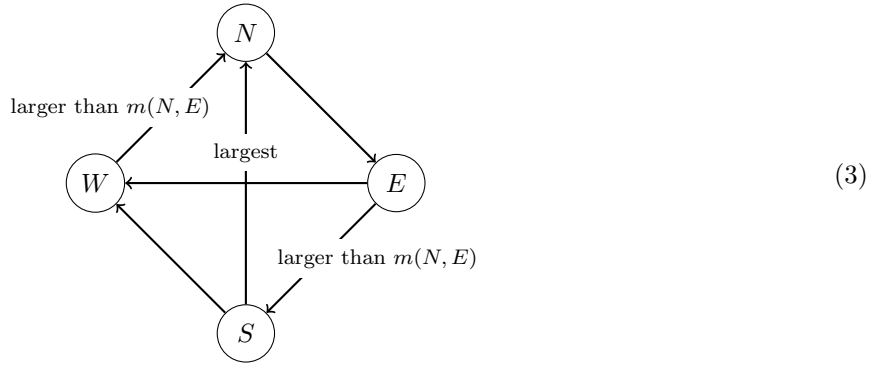
\begin{equation}
\begin{tikzpicture}[baseline={(current bounding box.center)}]

\node[circle,draw, minimum width=0.3in] at (0,0) (W) {$W$}; 
\node[circle,draw,minimum width=0.3in ] at (4,0) (E) {$E$}; 
\node[circle,draw,minimum width=0.3in ] at (2,2) (N) {$N$}; 
\node[circle,draw,minimum width=0.3in] at (2,-2) (S) {$S$};

\path[->,draw,  thick] (N) to node   { } (E);
\path[->,draw,thick] (E) to[pos=.75]  node   {} (W);
\path[->,draw,  thick] (S) to node {} (W);
\path[->,draw,thick] (S) to[pos=.625]  node[fill=white] {\footnotesize largest} (N);
\path[->,draw,  thick] (E) to node[fill=white]  {}  (S);
\path[->,draw,  thick] (W) to node {} (N);

\node[fill=white] at (.2,1) {\footnotesize larger than $m(N,E)$}; 

\node at (-.8,0) {};
\node[fill=white] at (3.75,-1) {\footnotesize larger than $m(N,E)$}; 
\end{tikzpicture}\label{FlipEdge}
\end{equation}
In (\ref{FlipEdge}),  $S$ dominates $W$ in wins, so $W$ cannot win by \textbf{Win Dominance}. By \textbf{Proximity to Condorcet}, neither $N$ nor $S$ can win, since they each have a loss larger than $E$'s single loss. Thus, $E$ must be the winner in (\ref{FlipEdge}). But this contradicts \textbf{Independence of Irrelevant Defeats}, since from (\ref{BoostEdges}) to (\ref{FlipEdge}) we only changed a margin that involves neither $W$ nor $E$. From this contradiction, we conclude that  $W$ is not the winner in a four cycle.

Now we claim that $N$ wins in the SL four cycle:

\begin{equation}
\begin{tikzpicture}[baseline={(current bounding box.center)}] 

\node[circle,draw, minimum width=0.3in] at (0,0) (W) {$W$}; 
\node[circle,draw,minimum width=0.3in ] at (4,0) (E) {$E$}; 
\node[circle,draw,minimum width=0.3in] at (2,2) (N) {$N$}; 
\node[circle,draw,minimum width=0.3in] at (2,-2) (S) {$S$}; 

\path[->,draw,  thick] (W) to node[fill=white] {\footnotesize smaller} (N);
\path[->,draw,  thick] (N) to node[fill=white]   {\footnotesize larger} (E);
\path[->,draw,thick] (E) to[pos=.75]  node   {} (W);
\path[->,draw,  thick] (S) to node {} (W);
\path[->,draw,thick] (N) to node {} (S);
\path[->,draw,  thick] (E) to node {} (S);

\node at (2,-3) {SL four cycle}; 

\end{tikzpicture}\label{SLEq}
\end{equation}
We already know that $W$ cannot win in (\ref{SLEq}). In addition, we know that $E$ cannot win by \textbf{Proximity to Condorcet}, since $N$'s single loss is smaller than $E$'s single loss. Now suppose for contradiction that $S$ is the winner. Consider this uniquely-weighted tournament, where $E$'s margin vs.~$W$ is the largest of all margins: 

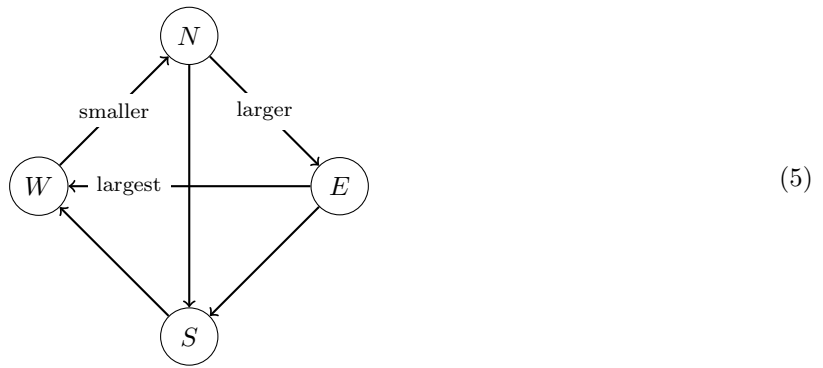
\begin{equation}
\begin{tikzpicture}[baseline={(current bounding box.center)}]  

\node[circle,draw, minimum width=0.3in] at (0,0) (W) {$W$}; 
\node[circle,draw,minimum width=0.3in ] at (4,0) (E) {$E$}; 
\node[circle,draw,minimum width=0.3in] at (2,2) (N) {$N$}; 
\node[circle,draw,minimum width=0.3in] at (2,-2) (S) {$S$}; 

\path[->,draw,  thick] (W) to node[fill=white] {\footnotesize smaller} (N);
\path[->,draw,  thick] (N) to node[fill=white]   {\footnotesize larger} (E);
\path[->,draw,thick] (E) to[pos=.75]  node[fill=white]   {\footnotesize largest} (W);
\path[->,draw,  thick] (S) to node {} (W);
\path[->,draw,thick] (N) to node {} (S);
\path[->,draw,  thick] (E) to node {} (S);

\end{tikzpicture}\label{BoostEW}
\end{equation}
In (\ref{BoostEW}), $W$ and $E$ cannot win, for the same reasons as before. But now $S$ cannot win by \textbf{Win Dominance}, since $E$ dominates $S$ in wins. Thus, $N$ must win in (\ref{BoostEW}). But this contradicts \textbf{Independence of Irrelevant Defeats}, since from (\ref{SLEq}) to (\ref{BoostEW}) we only changed a margin that involves neither $S$ nor $N$. From this contradiction, we conclude that $S$ cannot win in the SL four cycle, which leaves $N$ as the only possible winner, in agreement with the proposed method.

Finally, we claim that $E$ wins in the LS four cycle:

\begin{equation}
 \begin{tikzpicture}[baseline={(current bounding box.center)}]   

\node[circle,draw, minimum width=0.3in] at (0,0) (W) {$W$}; 
\node[circle,draw,minimum width=0.3in ] at (4,0) (E) {$E$}; 
\node[circle,draw,minimum width=0.3in ] at (2,2) (N) {$N$}; 
\node[circle,draw,minimum width=0.3in] at (2,-2) (S) {$S$}; 

\path[->,draw,  thick] (W) to node[fill=white] {\footnotesize larger} (N);
\path[->,draw,  thick] (N) to node[fill=white]   {\footnotesize smaller} (E);
\path[->,draw,thick] (E) to[pos=.75]  node   {} (W);
\path[->,draw,  thick] (S) to node {} (W);
\path[->,draw,thick] (N) to node {} (S);
\path[->,draw,  thick] (E) to node {} (S);

\node at (2,-3) {LS four cycle}; 

\end{tikzpicture}\label{LSEq}
\end{equation}
As above, $W$ cannot win in (\ref{LSEq}), and $N$ cannot win by \textbf{Proximity to Condorcet}, since $E$'s single loss is smaller than $N$'s single loss. Now suppose for contradiction that $S$ wins. Then we can strengthen $S$'s victory over $W$ by a sufficiently large $n$ such that when we also strengthen $N$'s victory over $E$ by $n$, we obtain a uniquely-weighted SL four cycle as in (\ref{SLEq}). By \textbf{Win Monotonicity}, $S$ must still win, but this contradicts the fact proven above that $S$ cannot win in an SL four cycle. From this contradiction, we conclude that $S$ cannot win in the LS four cycle, leaving $E$ as the only possible winner, in agreement with the proposed method.\end{proof}

We finish this section by proving that each of the axioms used for Theorem \ref{MainFour} is independent of the others.

\begin{proposition}\label{IndProp} For each axiom in Table~\ref{AxTable}, there is a weighted tournament solution for up to four candidates that violates that axiom while satisfying all the other axioms in Table~\ref{AxTable}.
\end{proposition}

\begin{proof} Table~\ref{AxTable} already shows that Win Dominance and Rare Ties are independent of the other axioms. 

To prove that Proximity to Condorcet is independent of the others, consider CGB$^+$ as defined in Remark~\ref{BordaNote}. It violates Proximity to Condorcet, since it chooses $E$ in the following weighted tournament in which $E$ has the greatest symmetric Borda score (of $16$, compared to $N$'s score of $14$), even though $N$ has the smallest loss and is therefore closer to being the Condorcet winner:
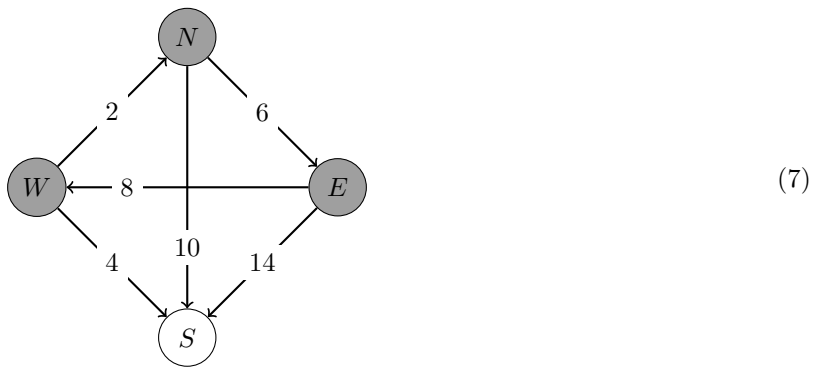
\begin{equation}
\begin{tikzpicture} [baseline={(current bounding box.center)}]    

\node[circle,draw, minimum width=0.3in, fill=gray!75] at (0,0) (W) {$W$}; 
\node[circle,draw,minimum width=0.3in, fill=gray!75] at (4,0) (E) {$E$}; 
\node[circle,draw,minimum width=0.3in, fill=gray!75] at (2,2) (N) {$N$}; 
\node[circle,draw,minimum width=0.3in] at (2,-2) (S) {$S$}; 

\path[->,draw,  thick] (W) to node[fill=white] {2} (N);
\path[->,draw,  thick] (N) to node[fill=white]   {6} (E);
\path[->,draw,  thick] (E) to[pos=.75]  node[fill=white]   {8} (W);
\path[->,draw,thick] (W) to node[fill=white] {4} (S);
\path[->,draw,thick] (N) to[pos=.75]  node[fill=white] {10} (S);
\path[->,draw,thick] (E) to node[fill=white] {14} (S);

\end{tikzpicture}
\end{equation}
\noindent However, CGB$^+$ satisfies Independence of Irrelevant Defeats, since a candidate's number of wins and sizes of wins and losses depend only on head-to-head matches involving that candidate, so changing head-to-head matches involving neither $A$ nor $B$ cannot change the winner from $A$ to $B$. It also satisfies Win Monotonicity, since increasing one of $A$'s wins by as much as one of $B$'s wins preserves $A$'s status as a Copeland winner, increases $A$'s sum of margins by as much as it does $B$'s, and it does not change $A$'s or $B$'s losses. Moreover, like any solution that chooses from among the Copeland winners, CGB$^+$ satisfies Win Dominance. Finally, it satisfies Rare Ties since in a uniquely-weighted tournament, there cannot be a tie in the sizes of losses.

To prove that Independence of Irrelevant Defeats is independent of the other axioms, consider the weighted tournament solution that we might call Uncovered-Minimax: select from the set  $\mathrm{Uncovered}(\mathsf{C},m)$ of  \textit{uncovered} candidates (recall Section~\ref{WinDom}) the one whose worst loss is smallest. Formally, the method outputs
\[\underset{X\in \mathrm{Uncovered}(\mathsf{C},m)}{\mathrm{argmin}} \mathrm{max} (\{ m(Y,X)\mid Y\in\mathsf{C}\mbox{ and } m(Y,X) > 0 \}).
 \]
Now consider the following LS four cycle, with uncovered candidates highlighted in blue:
 
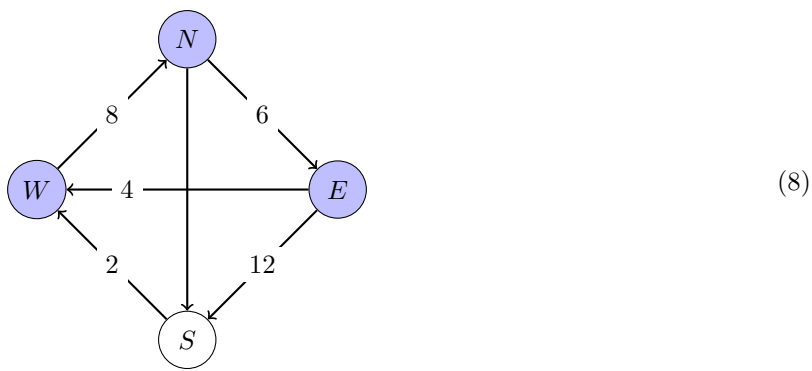
\begin{equation}
  \begin{tikzpicture} [baseline={(current bounding box.center)}]    

\node[circle,draw, minimum width=0.3in,fill=blue!25] at (0,0) (W) {$W$}; 
\node[circle,draw,minimum width=0.3in,fill=blue!25 ] at (4,0) (E) {$E$}; 
\node[circle,draw,minimum width=0.3in,fill=blue!25 ] at (2,2) (N) {$N$}; 
\node[circle,draw,minimum width=0.3in] at (2,-2) (S) {$S$}; 

\path[->,draw,  thick] (W) to node[fill=white] {8} (N);
\path[->,draw,  thick] (N) to node[fill=white]   {6} (E);
\path[->,draw,thick] (E) to[pos=.75]  node[fill=white]    {4} (W);
\path[->,draw,  thick] (S) to node[fill=white]  {2} (W);
\path[->,draw,thick] (N) to node {} (S);
\path[->,draw,  thick] (E) to node[fill=white] {12} (S);

\end{tikzpicture}\label{UMTourn}
\end{equation}
 In (\ref{UMTourn}), Uncovered-Minimax selects $W$, since $W$ is uncovered and its worst loss is smaller than that of any other candidate. But if we flip   $N$'s victory over $S$ to become a victory for $S$ over $N$, then $W$ will become covered by $S$, so Uncovered-Minimax will select $E$, since $E$ will then have the smallest loss among uncovered candidates. Thus, changing the irrelevant margin between $N$ and $S$ changes the winner from $W$ to $E$, thereby violating Independence of Irrelevant Defeats. 

On the other hand, Uncovered-Minimax satisfies all the other axioms. Suppose $A$ and $B$ are as in the statement of Proximity to Condorcet. If there is a Condorcet winner, then $B$ is covered and cannot be selected. If there is no Condorcet winner, then it follows from the assumption on $A$ in Proximity to Condorcet that $A$ has the most wins (possibly tied with others for most wins), so they cannot be covered, and their loss is smaller than $B$'s losses, so again $B$ cannot be selected. Hence Uncovered-Minimax satisfies Proximity to Condorcet. It also satisfies Win Monotonicity, since improving one of $A$'s wins by as much as one of $B$'s preserves $A$'s status as an uncovered candidate whose worst loss is smallest (note that no new candidates become uncovered). Of course Uncovered-Minimax satisfies Win Dominance, since if a candidate is not covered by any other candidate, then they are not dominated in wins by any other candidate. Finally, it satisfies Rare Ties since there cannot be a tie for the worst loss in a uniquely-weighted tournament.

Finally, we prove that Win Monotonicity is independent of the other axioms. To do so, we will use weighted tournaments of the following form, where the margin of $W$ vs.~$N$ is greater than 10 and the other margins are exactly as indicated:
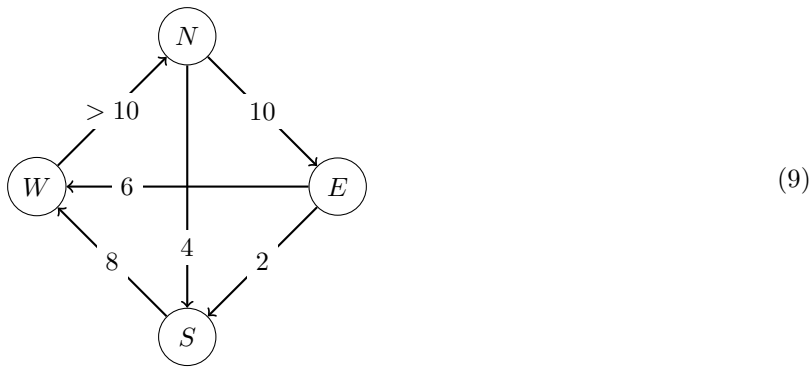
\begin{equation}
\begin{tikzpicture}[baseline={(current bounding box.center)}]    

\node[circle,draw, minimum width=0.3in] at (0,0) (W) {$W$}; 
\node[circle,draw,minimum width=0.3in ] at (4,0) (E) {$E$}; 
\node[circle,draw,minimum width=0.3in ] at (2,2) (N) {$N$}; 
\node[circle,draw,minimum width=0.3in ] at (2,-2) (S) {$S$}; 

\path[->,draw,  thick] (W) to node[fill=white] {$> 10$} (N);
\path[->,draw,  thick] (N) to node[fill=white]    {10}     (E);
\path[->,draw,thick] (E) to[pos=.75]  node[fill=white]   {6} (W);
\path[->,draw,  thick] (S) to node[fill=white] {8} (W);
\path[->,draw,thick] (N) to[pos=.75] node[fill=white]  {4} (S);
\path[->,draw,  thick] (E) to node[fill=white] {2} (S);
\end{tikzpicture}\label{GTourn}
\end{equation}
Define a weighted tournament solution $G$ as follows: if $T$ is a weighted tournament of the form in (\ref{GTourn}), then the output is $\{S\}$; otherwise the output agrees with the proposed method. Then $G$ clearly violates Win Monotonicity, since strengthening $S$'s margin vs.~$W$ in $T$ by $1$ while also strengthening $E$'s margin vs.~$W$ by $1$ results in a weighted tournament in which we pick the winner according to the proposed method, which is $E$. On the other hand, $G$ clearly satisfies Proximity to Condorcet, Win Dominance, and Rare Ties, since each of these axioms only concerns one weighted tournament at a time, and $G$ satisfies the axioms in $T$, and the proposed method satisfies the axioms in every other weighted tournament. 

It only remains to show that $G$ satisfies Independence of Irrelevant Defeats (IID). Since the proposed method does, it suffices to show that for each pair $\{X,Y\}$ of candidates not including $S$, if we start with a tournament $T$ as in (\ref{GTourn}) and then change the margin between $X$ and $Y$, then the output of the method cannot become $\{Z\}$ for some $Z\not\in \{S,X,Y\}$. There are three pairs of candidates to consider, namely $\{E,W\}$, $\{W,N\}$, and $\{N,E\}$. If we modify the margin of $E$ vs.~$W$ but keep it positive, then $\{E\}$ is the output of the proposed method, which is consistent with IID. On the other hand, if we flip $E$'s victory over $W$ to become a victory for $W$ over $E$, then $\{W\}$ is the output of the proposed method, which is again consistent with IID. Next, if we modify the margin of $W$ vs.~$N$ but keep it greater than $10$, then the output is still $\{S\}$, and if we keep the margin of $W$ vs.~$N$ positive but make it no greater than $10$, then the output of the proposed method is either $\{N,E\}$ or $\{N\}$, which is consistent with IID. On the other hand, if we flip $W$'s victory over $N$ to be a victory for $N$ over $W$, then $\{N\}$ is the output of the proposed method, which is again consistent with IID. Finally, if we modify $N$'s margin vs.~$E$ in any way, then one of $N$ or $E$ wins, which is again consistent with IID. This completes the proof that $G$ satisfies IID.\end{proof}

\section{Final Five}\label{Five}

Although Final Four elections have already been implemented in Alaska (but using Instant Runoff Voting instead of a Condorcet method), there is also support for Final Five elections, wherein five candidates compete in a general election under ranked voting, after having been selected in a preliminary process to compete in the general election (see  \citealt{Gehl2020}, \citealt{Gehl2023}). In this section, we axiomatically characterize an extension of the proposed method to elections with up to five candidates, by replacing one of our axioms from the previous section and bringing in one more.

The idea of the first axiom is that insofar as we think that a candidate with more head-to-head wins than any other candidate, i.e., a unique Copeland winner, deserves to be elected, then when there is a tie for the most head-to-head wins, we should elect the candidate who is \textit{closest} to being the unique Copeland winner, at least if it is clear who that candidate is:
\begin{itemize}
\item \textbf{Proximity to Copeland}: if $A$ can be made the unique Copeland winner by increasing one of $A$'s margins by some non-negative integer $n$, while $B$ cannot be made the unique Copeland winner by increasing \textit{every one} of $B$'s margins by $n$, then $F(T)\neq \{B\}$. 
\end{itemize}
All we have done is taken the Proximity to Condorcet axiom and replaced `Condorcet winner' with `unique Copeland winner'. The idea in both cases is the same: once you have committed to a certain kind of candidate---the Condorcet winner or a unique Copeland winner---being a clear winner, then when such a candidate does not exist, it is very natural to elect the candidate \textit{closest} to being such a clear winner, at least if someone clearly merits the distinction of ``closest.''\footnote{One might initially guess that Proximity to Copeland implies Proximity to Condorcet, since a Condorcet winner is a unique Copeland winner. However, the implication does not hold, as shown by the following weighted tournament solution: the score of a candidate $A$ is the least $n$ such that improving \textit{every one} of $A$'s margins by $n$ makes $A$ the unique Copeland winner; and the candidates with the lowest score win.  This solution satisfies Proximity to Copeland, since if $A$ can be made a unique Copeland winner by improving one of $A$'s margins by $n$, then $A$ can also be made a unique Copeland winner by improving every one of $A$'s margins by $n$, and if $B$ cannot be made a unique Copeland winner by improving every one of $B$'s margins by $n$, it follows that $A$'s score as defined above is strictly lower than $B$'s, so $B$ cannot be the winner. On the other hand, concerning Proximity to Condorcet, consider a top four-cycle as in Figure~\ref{NoUniqueCopeland} in which the three smallest losses in the whole weighted tournament belong to $Y$, but $Y$ also has the largest loss of all, which it suffers against one of the non-gray candidates in the four-cycle who has two losses. Then $Y$ has the best score as just defined, since by flipping $Y$'s three small losses, $Y$ will become the unique Copeland winner with three wins, as all other candidates will have at most two wins, and no other candidate can become the unique Copeland winner by such a small change of margin. However, the selection of $Y$ violates Proximity to Condorcet: a gray candidate in the four cycle can be made a Condorcet winner by improving one of their margins by some $n$ such that improving all of $Y$'s margins by that $n$ will not make $Y$ the Condorcet winner, due to $Y$ having the largest loss of any candidate. Despite the fact that Proximity to Copeland does not imply Proximity to Condorcet, all applications of Proximity to Condorcet in the proof of Theorem~\ref{MainFour} can be replaced by applications of Proximity to Copeland, for the reasons noted in the proof of Theorem~\ref{MainFive}.}

Admittedly, this new axiom simply takes for granted the election of the unique Copeland winner, when one exists. Thus, it will not convince those skeptical of electing the unique Copeland winner (though it is plausible that many voters will find electing a candidate who has more head-to-head wins than any other to be highly intuitive, based on analogies with sports). However, the purpose of the axiom here is not to convince such skeptics but rather to guide us to a natural generalization to Final Five elections of the proposed method from the previous section for Final Four elections.
And in this role, the axiom serves us well, since it helps distinguish the Final Five proposal briefly floated in \citealt{Holliday2025} from the Final Five proposal in this paper:
\begin{itemize}
\item Copeland-Global-Minimax (\citealt[\S~2.3]{Holliday2025}): select the candidate with the most head-to-head wins; if multiple candidates tie for the most wins, select the one \textit{whose worst head-to-head loss is smallest}.\footnote{The reason for the `Global' in the name of this method will be explained in Remark~\ref{LocalGlobal}.} Formally,
\[\mathrm{CGM}(\mathsf{C},m)=  \underset{X\in \mathrm{Copeland}(\mathsf{C},m)}{\mathrm{argmin}} \mathrm{max}( \{ m(Y,X)\mid Y\in\mathsf{C}\mbox{ and } m(Y,X) > 0 \}),
 \]
  where we set $\mathrm{max}(\varnothing)=0$.
\item Most Wins, Smallest Loss (this paper): select the candidate with the most head-to-head wins; if multiple candidates tie for the most wins, select the one \textit{who has the smallest head-to-head loss}. Formally,
\[\mathrm{MWSL}(\mathsf{C},m)=  \underset{X\in \mathrm{Copeland}(\mathsf{C},m)}{\mathrm{argmin}} \mathrm{min} (\{ m(Y,X)\mid Y\in\mathsf{C}\mbox{ and } m(Y,X) > 0 \}),
 \]
  where we set $\mathrm{min}(\varnothing)=0$.
\end{itemize}
These proposals lead to different results only in one extremely improbable type of Final Five election, shown in Figure \ref{Penta}, in which every candidate has exactly two wins and two losses (in every other five-candidate uniquely-weighted tournament, the Copeland winners have at most one loss). In Figure \ref{Penta}, suppose $A$ loses by $6$ and by $2$, and $B$ loses by $4$ and by $4$, and all other candidates lose by larger margins. Copeland-Global-Minimax declares $B$ the winner, since $B$'s worst loss of $4$ is smaller than $A$'s worst loss of $6$. By contrast, MWSL declares $A$ the winner, since $A$ has the smallest loss, namely a loss by $2$. Here $A$ is closest to being the unique Copeland winner, since we need only change their smallest loss of $2$ by a margin of $3$ to flip it to a win, and then they would be elected outright based on their number of wins, whereas changing anyone else's losses by a margin of $3$ would not change the win-loss records. Indeed, it is easy to see that MWSL satisfies Proximity to Copeland, whereas Copeland-Global-Minimax violates it, as shown by this example.

The practical advantage of MWSL over Copeland-Global-Minimax for Final Five elections is the greater simplicity of MWSL: we do not have to subject voters to the potentially confusing concept of ``the candidate whose worst loss is smallest.'' But the new proposal also has a philosophical advantage in that it can be motivated by the idea of closeness to a clear winner that we used before.

\begin{figure}[h!]
\begin{center}
\begin{tikzpicture}

\node[circle,draw] at (0,.1) (c) {$\quad\,$}; 
\node[circle,draw] at (4,.1) (a) {$\quad\,$}; 
\node[circle,draw] at (2,1.75) (b) {$\quad\,$}; 
\node[circle,draw] at (1,-2) (d) {$\quad\,$}; 
\node[circle,draw] at (3,-2) (e) {$\quad\,$};

\path[->,draw,thick] (b) to (a);
\path[->,draw,thick] (c) to (b);  
\path[->,draw,thick] (c) to (a);

\path[->,draw,thick] (d) to (c);  
\path[->,draw,thick] (d) to (b);

\path[->,draw,thick] (e) to (d);  
\path[->,draw,thick] (e) to (c);

\path[->,draw,thick] (a) to (e);  
\path[->,draw,thick] (a) to (d);

\path[->,draw,thick] (b) to (e);

  \end{tikzpicture}
\end{center}
\caption{The symmetric pentagram, which Harrison-Trainor \citeyearpar{HT2020} calls $T_{12}$.}\label{Penta}
\end{figure}
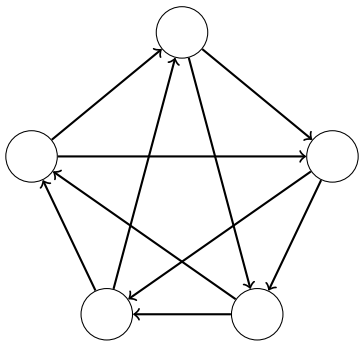

\begin{remark}\label{LocalGlobal}
As an aside, we note that like the proposed method from Section~\ref{Prelim}, the variant method from Section~\ref{Prelim} can be generalized in two different ways to Final Five elections, depending on whether we replace the $\iota$ in its definition by $\mathrm{max}$ or $\mathrm{min}$. We do not know of a good name for the generalization of the variant method using $\mathrm{min}$, but the generalization using $\mathrm{max}$ does have one:
\begin{itemize}
\item Copeland-Local-Minimax: select the candidate with the most head-to-head wins; if multiple candidates tie for the most wins, then among the candidates in that group, select the candidate who beats each of the others in that group head-to-head; if there is no such candidate, then select the candidate in that group whose worst loss to another candidate in that group is smallest. Formally,
\[\mathrm{CLM}(\mathsf{C},m)=  \underset{X\in \mathrm{Copeland}(\mathsf{C},m)}{\mathrm{argmin}}\; \mathrm{max}(\{ m(Y,X)\mid Y\in \mathrm{Copeland}(\mathsf{C},m)\mbox{ and } m(Y,X) > 0 \}).
 \]
\end{itemize}
We call this `Copeland-Local-Minimax' because it is equivalent to taking the given weighted tournament $T=(\mathsf{C},m)$, deleting from it all candidates who are not Copeland winners to obtain a weighted tournament $T^{-}$ and then applying Minimax ``locally'' to $T^{-}$. By contrast, Copeland-Global-Minimax as defined above chooses the Copeland winners who have the best Minimax score computed ``globally'' with respect to the entire tournament $T$. For more on this local vs.~global distinction, which also induces a distinction between Copeland-Local-Borda and Copeland-Global-Borda (recall Remark~\ref{BordaNote}), see \citealt[\S~6.3]{HP2021}.
\end{remark}

The second axiom we will add in the context of Final Five elections is one of the original axioms used to motivate the proposed method in Final Four elections (see \citealt[\S~4.2]{Holliday2025}), namely Immunity to Spoilers (see \citealt{HP2023}). Here we state the axiom in a single-winner fashion, since we are already assuming Rare Ties and uniquely-weighted tournaments:
\begin{itemize}
\item \textbf{Immunity to Spoilers}: if $F(T_{-B})=\{A\}$, and $A$ beats $B$ head-to-head in $T$, then for all $C$ distinct from $A$ and $B$, we have $F(T)\neq \{C\}$.
\end{itemize}
Here $T_{-B}$ is the weighted tournament that results from deleting $B$ from $T$.\footnote{Formally, where $T=(\mathsf{C},m)$ and $B\in\mathsf{C}$, define $T_{-B}=(\mathsf{C}',m')$ where $\mathsf{C}'=\mathsf{C}\setminus \{B\}$ and $m'$ is the restriction of $m$ to $\mathsf{C}'\times\mathsf{C}'$.} The axiom says that if $A$ would win without $B$ in the election, and $A$ beats $B$ head-to-head when $B$ is included, then it should not happen that when $B$ is included, both $A$ and $B$ lose; for then $B$ is exactly the kind of \textit{spoiler} that makes voters worried about ``spoiler effects.'' As noted by Holliday and Pacuit \citeyearpar[\S~1.1]{HP2023}, the assumption that $A$ beats $B$ head-to-head is essential here, for of course if $A$ \textit{loses} head-to-head to $B$, then it is no wonder that $B$'s entrance into the election may cause $A$ to lose; and in any case no reasonable voting method can satisfy the strengthened version of the axiom that removes the assumption that $A$ beats~$B$.\footnote{Suppose that without $B$ in the election, $A$ beats $C$ and therefore wins in the two-candidate election with $\{A,C\}$. Further suppose that with $B$ included in the election, $B$ beats $A$ by a larger margin than that by which $A$ beats $C$. Indeed, suppose $C$ has the smallest loss (and, if one wishes, the largest win) of any of the three candidates. Then $C$ should win in the election with $\{A,B,C\}$, but this violates the strengthened version of the axiom that drops the assumption that $A$ beats $B$.}

Both Copeland and Minimax satisfy Immunity to Spoilers, since if $A$ has the most wins (resp.~smallest worst loss) without $B$ in the election, and $A$ beats $B$ head-to-head, then with $B$ included, if $B$ does not have the most wins (resp.~smallest worst loss), then $A$ still does, so some third candidate $C$ cannot suddenly win. Similarly, the proposed method and its generalizations as MWSL and CGM satisfy Immunity to Spoilers, since if $A$ has the smallest loss (resp.~smallest worst loss) among candidates with the most wins without $B$ in the election, then with $B$ included, if $B$ does not have the smallest loss (resp.~smallest worst loss) among candidates with the most wins, then $A$ still does, so some third candidate cannot suddenly win. By contrast, the axiom is violated by the variant method defined in Section~\ref{Prelim} and hence by the CLM method defined in Remark~\ref{LocalGlobal}, as shown in Figure \ref{ISViolation}: $E$ wins without $S$ in the election, and $E$ beats $S$ head-to-head, but with $S$ included in the election, both $E$ and $S$ lose and $N$ wins, so $S$ is a spoiler. Incidentally, this example also shows that the Borda-based refinements of Copeland from Remark \ref{BordaNote} violate Immunity to Spoilers: they select $E$ without $S$ in the election, but if we make $N$'s margin over $S$ sufficiently large compared to $E$'s margin over $S$, then the Borda-based methods will select $N$, so again $S$ is a spoiler.

\begin{figure}[h!]
\begin{center}
 \begin{minipage}{2.5in}
 \begin{center}\begin{tikzpicture} 

\node[circle,draw, minimum width=0.3in] at (0,0) (W) {$W$}; 
\node[circle,draw,minimum width=0.3in] at (4,0) (E) {$E$}; 
\node[circle,draw,minimum width=0.3in] at (2,2) (N) {$N$}; 

\path[->,draw, thick] (W) to node[fill=white] {\footnotesize medium} (N);
\path[->,draw, thick] (N) to node[fill=white]   {\footnotesize small} (E);
\path[->,draw, thick] (E) to node[fill=white]   {\footnotesize large}  (W);

\node at (2,-1) {three cycle}; 

\end{tikzpicture}\end{center}\end{minipage} \begin{minipage}{2.5in}
 \begin{center}\begin{tikzpicture} 

\node[circle,draw, minimum width=0.3in] at (0,0) (W) {$W$}; 
\node[circle,draw,minimum width=0.3in] at (4,0) (E) {$E$}; 
\node[circle,draw,minimum width=0.3in] at (2,2) (N) {$N$}; 
\node[circle,draw,minimum width=0.3in] at (2,-2) (S) {$S$}; 

\path[->,draw, thick] (W) to node[fill=white] {\footnotesize medium} (N);
\path[->,draw, thick] (N) to node[fill=white]   {\footnotesize small} (E);
\path[->,draw, thick] (E) to[pos=.74]  node[fill=white]   {\footnotesize large}  (W);
\path[->,draw, thick] (S) to node {} (W);
\path[->,draw, thick] (N) to node {} (S);
\path[->,draw, thick] (E) to node {} (S);

\node at (2,-3) {LS four cycle}; 

\end{tikzpicture}\end{center}\end{minipage}

\end{center}
\caption{A violation of Immunity to Spoilers by the variant method: on the left, the variant method selects $E$, but on the right, where we add a loser $S$ who loses head-to-head to $E$, the variant method changes to selecting $N$ (since now the Copeland winners are only $N$ and $E$, and $N$ beats $E$), so $S$ is a spoiler.}\label{ISViolation}
\end{figure}
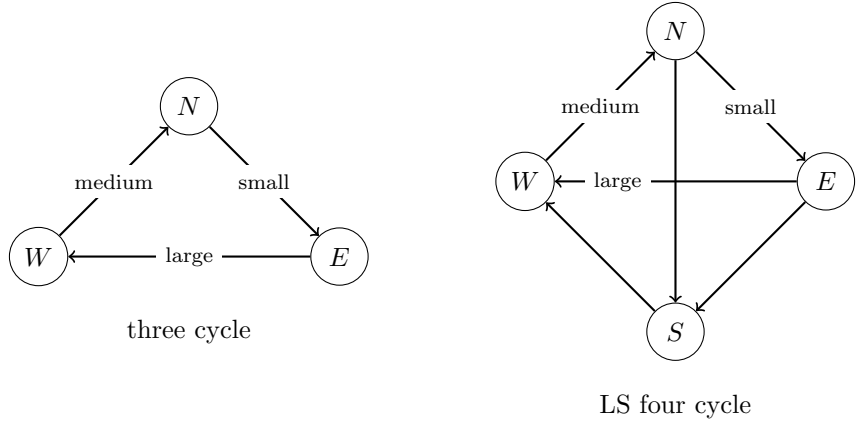

As an aside, it is noteworthy that with the help of Immunity to Spoilers, we can prove the Condorcet Criterion: the Condorcet winner must be the unique winner, whenever the Condorcet winner exists.

\begin{proposition}\label{CondorcetDerived} Let $F$ be a weighted tournament solution on the domain of all uniquely-weighted tournaments with a finite number of candidates. Assume that $F$ satisfies Immunity to Spoilers and Rare Ties. Further assume that in any two-candidate tournament, $F$ selects only the candidate who beats the other, following majority rule. Then $F$ satisfies the Condorcet Criterion for any number of candidates.\end{proposition}

\begin{proof} We prove by induction on the number of candidates that in any uniquely-weighted tournament $T$, if $T$ has a Condorcet winner $A$, then $F(T)=\{A\}$. The case of two candidates is given by the majority rule assumption. Now suppose $T$ has more than two candidates, and $A$ is the Condorcet winner. Suppose for contradiction that $F(T)\neq\{A\}$, so by Rare Ties, $F(T)=\{C\}$ for some $C$ distinct from $A$. Since $T$ has more than two candidates, there is some $B$ distinct from both $A$ and $C$. Now in $T_{-B}$, $A$ is still the Condorcet winner, so by the inductive hypothesis, $F(T_{-B})=\{A\}$. Moreover, since $A$ is the Condorcet winner in $T$, $A$ beats $B$ head-to-head in $T$. Then by Immunity to Spoilers, $F(T)\neq \{C\}$, which contradicts what we derived above. Thus, we conclude that $F(T)=\{A\}$, which proves that $F$ satisfies the Condorcet Criterion.\end{proof}

An updated version of Table~\ref{AxTable} including the two new axioms of this section is given in Table~\ref{AxTable2}.

\begin{table}[h]
\begin{center}
\begin{tabular}{l|c|c|c|c|c}
Axiom & Copeland & Minimax & MWSL & CGM & CLM \\
\hline
Proximity to Condorcet  & $\checkmark$ & $\checkmark$ & $\checkmark$& $\checkmark$  & --  \\
Proximity to Copeland  & $\checkmark$ & -- & $\checkmark$& --  & --  \\
Immunity to Spoilers & $\checkmark$ & $\checkmark$ & $\checkmark$ &$\checkmark$& -- \\
Independence of Irrelevant Defeats & $\checkmark$ & $\checkmark$ & $\checkmark$ &$\checkmark$& --  \\
Win Monotonicity & $\checkmark$ & $\checkmark$ & $\checkmark$ &$\checkmark$& $\checkmark$    \\
Win Dominance & $\checkmark$ & -- & $\checkmark$&$\checkmark$& $\checkmark$ \\
Rare Ties & -- & $\checkmark$ & $\checkmark$&$\checkmark$& $\checkmark$  \\
\end{tabular}
\end{center}
\caption{Satisfaction ($\checkmark$) and violation (--) of axioms by voting methods. `MWSL' stands for Most Wins, Smallest Loss; `CGM' stands for Copeland-Global-Minimax; and `CLM' stands for Copeland-Local-Minimax.}\label{AxTable2}
\end{table}

Let us now prove that with these axioms, we can completely characterize MWSL for uniquely-weighted tournaments with up to five candidates.

\begin{theorem}\label{MainFive} Let $F$ be a weighted tournament solution for up to five candidates satisfying all the axioms from Section~\ref{Axioms} but with Proximity to Condorcet replaced by Proximity to Copeland, plus Immunity to Spoilers. Then in any uniquely-weighted tournament, $F$ selects the same winner as~MWSL.\end{theorem}

\begin{proof} As before, if $T$ is uniquely-weighted, then \textbf{Rare Ties} implies that $F$ selects a unique winner in $T$. 

The case of four or fewer candidates is covered by Theorem \ref{MainFour}, since in the proof of that theorem, all applications of Proximity to Condorcet can be replaced by applications of \textbf{Proximity to Copeland}: Proximity to Condorcet is only used to show that (i) a Condorcet winner must be the winner, when one exists, and (ii) when two candidates both have exactly two wins and one loss, then the one with the larger loss cannot be the winner, and  \textbf{Proximity to Copeland} also delivers these results. 

 For five candidates, if there is a unique Copeland winner, then $F$ must select the unique Copeland winner by the $n=0$ case of \textbf{Proximity to Copeland}. There are only five tournaments (up to isomorphism) without a unique Copeland winner (see, e.g., Box VIII and  Table 3 of \citealt{HT2020} or the Appendix of \citealt{Moon1968}), shown in Figures \ref{Penta} and~\ref{NoUniqueCopeland}. In any uniquely-weighted version of the pentagram in Figure \ref{Penta} in which every candidate has two losses, the candidate with the smallest loss wins by \textbf{Proximity to Copeland}. 

For any uniquely-weighted tournament $T$ whose underlying tournament falls into one of the classes in Figure \ref{NoUniqueCopeland}, the argument is as follows. Let $A$ be whichever of the gray candidates (the Copeland winners) has the smallest loss among gray candidates, so MWSL selects $A$. If we remove $Y$, then $T_{-Y}$ is a top cycle (if $T$ is the top top-cycle) or four cycle (in the other three cases) in which $A$ is the winner by the proof of Theorem~\ref{MainFour}. Then since $A$ beats $Y$ head-to-head, it follows by \textbf{Immunity to Spoilers} that either $A$ still wins in $T$, in agreement with MWSL, or $Y$ wins. But $Y$ does not win in the top top-cycle or the top four-cycle by \textbf{Win Dominance}, since every other candidate dominates $Y$ in wins. So $A$ wins in the top top-cycle and the top four-cycle. It remains to show that $Y$ does not win in the mid-cycle order or the gyroscope.

In the mid-cycle order, suppose for contradiction that $Y$ wins. Then strengthen the margin of $Z$ vs.~$V$ so that it is greater than any other margin (thereby preserving the uniquely-weighted property), including that of $Y$ vs.~$V$. In the new $T'$ obtained in this way, $Z$ dominates $Y$ in wins, since the only candidate $Y$ beats is $V$. Hence $Y$ cannot win by \textbf{Win Dominance}, so $A$ must win by our reasoning using \textbf{Immunity to Spoilers} above (since by Theorem \ref{MainFour}, $A$ is the winner in $T'_{-Y}$, just as in $T_{-Y}$); but this contradicts \textbf{Independence of Irrelevant Defeats}, since the margin of $Z$ vs.~$V$ does not involve $A$ or $Y$. Thus, we conclude that $Y$ does not win in our original $T$, so $A$ does by our reasoning using \textbf{Immunity to Spoilers} above.

In the gyroscope, let $Z$ be the gray candidate who is not $A$. Then we can prove that $Y$ does not win exactly as in the previous paragraph except that $T$ is the gyroscope instead of the mid-cycle order.\end{proof}

\begin{figure}[h!]
\begin{center}
\begin{minipage}{2.5in}\begin{center}\begin{tikzpicture}
\node[rectangle,draw, minimum width = 4cm, minimum height = 2.5cm] at (1.5,.75) (rect1) {};
\node[rectangle,draw, minimum width = 4.25cm, minimum height = 4.25cm] at (1.5,0) (rect2) {};

\node[circle,draw,fill=gray!75] at (0,0)      (a) {$\quad\,$}; 
\node[circle,draw,fill=gray!75 ] at (3,0)      (c) {$\quad\,$}; 
\node[circle,draw,fill=gray!75 ] at (1.5,1.5)  (b) {$\quad\,$}; 
\node[circle,draw ] at (1.5,-1.5) (d) {$\quad\,$};
\node[circle,draw,fill=red!50 ] at (1.5,-3) (e) {$Y$};

\path[->,draw,thick] (a) to  (b);
\path[->,draw,thick] (b) to  (c);
\path[->,draw,thick] (c) to (a);
\path[->,draw,thick] (rect1) to  (d);
\path[->,draw,thick] (rect2) to  (e);

\end{tikzpicture}\\
top top-cycle ($T_4$)\\
\end{center}\vspace{.25in}\end{minipage}\begin{minipage}{2.5in}\begin{center}\begin{tikzpicture}
\node[circle,draw,fill=gray!75 ] at (0,-1.5)      (a) {$\quad\,$}; 
\node[circle,draw ] at (3,-1.5)      (c) {$\quad\,$}; 

\node[circle,draw,fill=gray!75 ] at (1.5,1.5) (d) {$\quad\,$};
\node[circle,draw ] at (1.5,.5) (b) {$\quad\,$};

\node[rectangle,draw, minimum width = 1.25cm, minimum height = 2cm] at (1.5,1) (rect1) {};

\path[<-,draw,thick] (a) to (c);
\path[<-,draw,thick] (c) to  (rect1);
\path[<-,draw,thick] (rect1) to  (a);
\path[->,draw,thick] (d) to (b);

\node[circle,draw,fill=red!50 ] at (1.5,-3) (e) {$Y$};
\node[rectangle,draw, minimum width = 4.25cm, minimum height = 4.25cm] at (1.5,0) (rect2) {};
\path[->,draw,thick] (rect2) to  (e);

\end{tikzpicture}\\
top four-cycle ($T_6$)\\
\end{center}\vspace{.25in}
\end{minipage}

\begin{minipage}{2.5in}\begin{center}\begin{tikzpicture}
\node[circle,draw,fill=gray!75 ] at (1.5,.5) (d) {$X$};
\node[circle,draw ] at (1.5,-1)      (c) {$Z$}; 
\node[circle,draw,fill=red!50  ] at (1.5,-2.5)      (a) {$Y$};

\node[circle,draw ] at (3,.5)  (b) {$V$}; 
\node[circle,draw,fill=gray!75 ] at (4.5,-1) (e) {$\quad\,$};

\node[rectangle,draw, minimum width = 1.25cm, minimum height = 4.25cm] at (1.5,-1) (rect1) {};
\node[rectangle,draw, minimum width = 1cm, minimum height = 2.5cm] at (1.5,-.25) (rect2) {};

\path[->,draw,thick] (rect2) to  (a);
\path[->,draw,thick] (d) to  (c);

\path[<-,draw,thick] (b) to  (rect1);
\path[<-,draw,thick] (rect1) to  (e);
\path[<-,draw,thick] (e) to  (b);

\end{tikzpicture}\\
mid-cycle order ($T_7$)\\
\end{center}
\end{minipage}\begin{minipage}{2.5in}
\begin{center}\begin{tikzpicture}
\node[circle,draw,fill=gray!75 ] at (0,0)      (a) {$X$}; 
\node[circle,draw,fill=gray!75 ] at (4,0)      (c) {$\quad\,$}; 
\node[circle,draw,fill=red!50 ] at (2,2)  (b) {$Y$}; 
\node[circle,draw] at (2,-2) (d) {$V$};
\node[circle,draw ] at (1.333,.666) (e) {$\quad\,$};

\path[->,draw,thick,bend left] (a) to  (b);
\path[<-,draw,thick,bend left] (b) to  (c);
\path[->,draw,thick,bend left] (c) to (d);
\path[<-,draw,thick,bend left] (d) to (a);
\path[<-,draw,thick,bend right] (a) to  (c);
\path[->,draw,thick,bend left] (b) to  (d);

\path[<-,draw,thick,bend right=20] (e) to  (a);
\path[->,draw,thick,bend left=20] (e) to  (b);
\path[<-,draw,thick,bend right=20] (e) to  (d);
\path[<-,draw,thick,bend right=20] (c) to  (e);

\end{tikzpicture}\\
gyroscope ($T_8$)\\
\end{center}
\end{minipage}

\end{center}
\caption{Tournaments for five candidates without a unique Copeland winner, in addition to the one in Figure~\ref{Penta}. An arrow from a candidate $C$ to a rectangle indicates that there are arrows from $C$ to \textit{every candidate} in the rectangle. Similarly, an arrow from a rectangle to a candidate $C$ indicates that there are arrows from \textit{every candidate} in the rectangle to $C$. The parenthetical names for tournaments match those in \citealt{HT2020}.}\label{NoUniqueCopeland}
\end{figure}

Here we make no claim of independence like Proposition~\ref{IndProp}, since there may be ways to replace the applications of Immunity to Spoilers in the proof of Theorem~\ref{MainFive} with applications of other axioms, along the lines of the proof of Theorem~\ref{MainFour}.\footnote{For example, consider the top top-cycle in Figure~\ref{NoUniqueCopeland}. $Y$ cannot win by Win Dominance. But neither can the candidate immediately above $Y$ in the diagram in Figure~\ref{NoUniqueCopeland}, whom we will call $Z$. Suppose for contradiction that $Z$ does win in $T$. Then let $B$ be the gray candidate with the greatest loss, and increase $B$'s margin vs.~$Y$ so that it is larger than $Z$'s margin vs.~$Y$, while maintaining the uniquely-weighted property. Then $Z$ cannot win in the new $T'$ by Win Dominance. But also, by Proximity to Copeland,  $B$ cannot win in $T'$, since the gray candidate with the smallest loss is closer to being the unique Copeland winner than $B$ is. Hence some third candidate must win after we increase $B$'s margin vs.~$Y$. But this contradicts Independence of Irrelevant Defeats. Thus, $Z$ cannot win in $T$. It follows by Proximity to Copeland that the only candidate remaining who can win in $T$ is the gray candidate with the smallest loss. In this way, we can replace the application of Immunity to Spoilers in the part of the proof of Theorem \ref{MainFive} concerning the top top-cycle with applications of the other axioms. Whether this is possible in all cases of the proof is a question we leave open.} But as Immunity to Spoilers was one of the original motivating axioms for MWSL in the four-candidate case (see \citealt{Holliday2025}), it seems reasonable from a normative point of view to use it to reduce much of the five-candidate case to the four-candidate case. We leave for another occasion the logical question of whether Immunity to Spoilers is independent of the other axioms, as well as the problem of characterizing MWSL in six-candidate elections and beyond.

\section{Conclusion}\label{Conclusion}

We have provided an axiomatic characterization of the simple Condorcet voting method for Final Four elections from \citealt{Holliday2025}, as well as an axiomatic characterization of a new generalization of that method to Final Five elections. These are important steps toward understanding the normative principles that might help justify the use of these methods---alongside considerations of their simplicity---in Final Four or Final Five political elections. But they are far from the end of the story. Voting methods can be assessed not only in terms of the axioms they satisfy or violate but also in terms of their probability of violating axioms, the representativeness of their winners, the extent to which they mitigate incentives for strategic voting, and more. A full case for the use of Condorcet methods, including the simple Condorcet method studied here, must take into account all of these approaches to the normative assessment of methods of voting.

\subsection*{Acknowledgments}

I am grateful for discussions with Ned Foley that motivated me to address the differences between the proposed method and the variant method and for discussions with Robbie Robinette that motivated me to address generalizations of the proposed method to Final Five elections. I thank Milan Moss\'{e}, Eric Pacuit, and Nic Tideman for helpful comments on this paper.

\bibliographystyle{plainnat}
\bibliography{social}

\end{document}